\documentclass[submission,copyright,creativecommons]{eptcs}
\providecommand{\event}{QPL 2020} 
\usepackage{underscore}           

\usepackage[utf8]{inputenc}
\usepackage{amssymb}
\usepackage{amsthm}
\usepackage{float}
\usepackage{esvect}
\usepackage{stmaryrd}
\usepackage{etoolbox}
\usepackage{subcaption}
\usepackage{mathtools}
\usepackage[dvipsnames]{xcolor}
\usepackage{tikz}
\usetikzlibrary{calc}
\usetikzlibrary{shapes.geometric}
\usetikzlibrary{shapes.misc}
\usetikzlibrary{positioning}
\usetikzlibrary{external}

\title{Quantum Algorithms and Oracles with the Scalable ZX-calculus}

\author{Titouan Carette \quad Yohann D'Anello \quad Simon Perdrix\\Université de Lorraine, CNRS, Inria, LORIA\\ F 54000 Nancy, France}

\def\titlerunning{Quantum Algorithms and Oracles with the Scalable ZX-calculus}
\def\authorrunning{T. Carette, Y. D'Anello \& S. Perdrix}

\newcommand{\quot}[2]{{\raisebox{.2em}{$#1$}\left/\raisebox{-.2em}{$#2$}\right.}}
\newcommand{\interp}[1]{\left\llbracket#1\right\rrbracket}

\newcommand{\bvdots}{ \tikz[baseline, every node/.style={inner sep=0}]{ \node at (0,0){.}; \node at (0,-6pt){.}; \node at (0,6pt){.}; } }
\newcommand{\ket}[1]{|#1\rangle}
\newcommand{\bra}[1]{\langle#1|}
\newcommand{\df}{\stackrel{\scriptscriptstyle def}{=}}
\usepackage{tikzit}

\let\oldtikzfig\tikzfig
\renewcommand{\tikzfig}[1]{
	\tikzsetnextfilename{#1}
	\oldtikzfig{#1}
}

\let\oldctikzfig\ctikzfig
\renewcommand{\ctikzfig}[1]{
	\tikzsetnextfilename{#1}
	\oldctikzfig{#1}
}
\definecolor{zx_grey}{RGB}{211,211,211}

\tikzstyle{gn}=[fill=green, draw=black, shape=circle, tikzit category=ZX, tikzit fill=green, tikzit draw=black, tikzit shape=circle, inner sep=0.1em]
\tikzstyle{rn}=[fill=red, draw=black, shape=circle, tikzit fill=red, tikzit draw=black, tikzit category=ZX, tikzit shape=circle, inner sep=0.1em]
\tikzstyle{divide}=[regular polygon, regular polygon sides=3, shape border rotate=90, draw=black, fill={zx_grey}, inner sep=1.5pt, tikzit category=scal, rounded corners=0.8mm]
\tikzstyle{black}=[fill=black, draw=black, shape=circle, tikzit fill=black, tikzit draw=black, tikzit shape=circle, tikzit category=IH, inner sep=2pt]
\tikzstyle{gather}=[fill={zx_grey}, draw=black, tikzit category=scal, rounded corners=0.8mm, regular polygon, regular polygon sides=3, shape border rotate=-90, inner sep=1.5pt]
\tikzstyle{ggen}=[fill=white, draw=black, shape=rectangle, rounded corners=2mm, line width=1pt, tikzit draw=red, tikzit category=scal]
\tikzstyle{white}=[fill=white, draw=black, shape=circle, inner sep=2pt, tikzit category=IH]
\tikzstyle{mbox}=[fill=white, draw=black, rounded rectangle, rounded rectangle west arc=none, tikzit category=scal, tikzit shape=rectangle]
\tikzstyle{A}=[fill=white, shape=circle, tikzit category=scal, inner sep=1pt]
\tikzstyle{ggreen}=[fill=green, draw=black, shape=circle, tikzit category=SZX, tikzit fill=green, tikzit draw=black, line width=1pt, inner sep=0.1em]
\tikzstyle{gred}=[fill=red, draw=black, shape=circle, rounded corners=2mm, tikzit category=SZX, inner sep=0.1em, tikzit fill=red, line width=1pt]
\tikzstyle{ghad}=[minimum size=3mm, font={\scriptsize\boldmath}, shape=rectangle, inner sep=1mm, line width=1pt, outer sep=-1.5mm, scale=0.8, tikzit shape=rectangle, draw=black, fill=yellow, tikzit draw=blue]
\tikzstyle{boxm}=[fill=white, draw=black, rounded rectangle, tikzit category=scal, tikzit shape=rectangle, rounded rectangle east arc=none]
\tikzstyle{box}=[fill=white, draw=black, shape=rectangle]
\tikzstyle{had}=[fill=yellow, draw=black, shape=rectangle, tikzit category=ZX, tikzit fill=yellow, tikzit draw=black, inner sep=2.5pt]
\tikzstyle{gwhite}=[fill=white, draw=black, shape=circle, tikzit fill=white, tikzit shape=circle, line width=1 pt, inner sep=2 pt, tikzit draw=red]
\tikzstyle{gblack}=[fill=black, draw=black, shape=circle, tikzit fill=black, tikzit shape=circle, line width=1 pt, inner sep=2 pt, tikzit draw=red]
\tikzstyle{antipode}=[fill=red, draw=black, shape=rectangle, tikzit fill=red, tikzit draw=black, tikzit shape=rectangle, inner sep=2pt]
\tikzstyle{diamond}=[fill=white, draw=black, shape=diamond, inner sep=2pt]
\tikzstyle{mongr}=[fill=green, draw=green, shape=circle, inner sep=2pt]
\tikzstyle{monbl}=[fill=blue, draw=blue, shape=circle, inner sep=2pt]
\tikzstyle{bg}=[inner sep=0.7mm, minimum width=0pt, minimum height=0pt, fill=green, draw=white, very thick, shape=circle]
\tikzstyle{br}=[inner sep=0.7mm, minimum width=0pt, minimum height=0pt, fill=red, draw=white, very thick, shape=circle]
\tikzstyle{rmat}=[draw, signal, fill=red, signal to=east, signal from=west, inner sep=1pt, minimum height=6pt]
\tikzstyle{lmat}=[draw, signal, fill=red, signal to=west, signal from=east, inner sep=1pt, minimum height=6pt]
\tikzstyle{umat}=[draw, signal, fill=red, signal to=north, signal from=south, inner sep=1pt, minimum width=6pt]
\tikzstyle{dmat}=[draw, signal, fill=red, signal to=south, signal from=north, inner sep=1pt, minimum width=6pt]
\tikzstyle{box}=[shape=rectangle, text height=1.5ex, text depth=0.25ex, yshift=0.5mm, fill=white, draw=black, minimum height=5mm, yshift=-0.5mm, minimum width=5mm, font={\small}]
\tikzstyle{Z dot}=[inner sep=0mm, minimum size=2mm, shape=circle, draw=black, fill={rgb,255: red,160; green,255; blue,160}]
\tikzstyle{gdot}=[minimum size=3mm, font={\scriptsize\boldmath}, shape=rectangle, rounded corners=1.3mm, inner sep=1mm, outer sep=-1.8mm, scale=0.8, tikzit shape=circle, draw=black, fill=green, tikzit draw=blue]
\tikzstyle{X dot}=[Z dot, shape=circle, draw=black, fill={rgb,255: red,220; green,0; blue,0}]
\tikzstyle{rdot}=[minimum size=3mm, font={\scriptsize\boldmath}, shape=rectangle, rounded corners=1.3mm, inner sep=1mm, outer sep=-1.8mm, scale=0.8, tikzit shape=circle, draw=black, fill=red, tikzit draw=blue]
\tikzstyle{grdot}=[minimum size=3mm, font={\scriptsize\boldmath}, shape=rectangle, rounded corners=1.3mm, inner sep=1mm, line width=1pt, outer sep=-1.5mm, scale=0.8, tikzit shape=circle, draw=black, fill=red, tikzit draw=blue]
\tikzstyle{ggdot}=[minimum size=3mm, font={\scriptsize\boldmath}, shape=rectangle, line width=1pt, rounded corners=1.3mm, inner sep=1mm, outer sep=-1.5mm, scale=0.8, tikzit shape=circle, draw=black, fill=green, tikzit draw=blue]
\tikzstyle{arrow}=[-->]
\tikzstyle{rfarr}=[draw, signal, fill=black, signal to=east, signal from=west, inner sep=1pt, minimum height=6pt]
\tikzstyle{lfarr}=[draw, signal, fill=black, signal to=west, signal from=east, inner sep=1pt, minimum height=6pt]
\tikzstyle{ufarr}=[draw, signal, fill=black, signal to=north, signal from=south, inner sep=1pt, minimum width=6pt]
\tikzstyle{dfarr}=[draw, signal, fill=black, signal to=south, signal from=north, inner sep=1pt, minimum width=6pt]
\tikzstyle{ry}=[draw, signal, fill=yellow, signal to=east, signal from=west, inner sep=1pt, minimum height=6pt]
\tikzstyle{ly}=[draw, signal, fill=yellow, signal to=west, signal from=east, inner sep=1pt, minimum height=6pt]
\tikzstyle{uy}=[draw, signal, fill=yellow, signal to=north, signal from=south, inner sep=1pt, minimum width=6pt]
\tikzstyle{dy}=[draw, signal, fill=yellow, signal to=south, signal from=north, inner sep=1pt, minimum width=6pt]

\tikzstyle{arrow}=[->]
\tikzstyle{very thick}=[-, line width=1pt, tikzit draw=red]
\tikzstyle{reprise}=[-, line width=2pt, tikzit draw=blue]
\tikzstyle{pointille}=[dashed, -]
\tikzstyle{red}=[-, draw=red]
\tikzstyle{blue}=[-, draw=blue]
\tikzstyle{green}=[-, draw=green]
\tikzstyle{strike}=[-, tikzit draw={rgb,255: red,191; green,0; blue,64}, strike through]
\tikzstyle{strike'}=[-, tikzit draw=cyan, strike bend]
\tikzstyle{dashed arrow}=[->, tikzit draw=green, draw=black, dashed]

\AtBeginEnvironment{pmatrix}{\setlength{\arraycolsep}{2.5pt} \renewcommand{\arraystretch}{1}}

\newtheorem{lemma}{lemma}

\usetikzlibrary{circuits.ee.IEC}
\usetikzlibrary{shapes.gates.logic.US,shapes.gates.logic.IEC}

\newcommand{\ground}{%
	\begin{tikzpicture}[circuit ee IEC,yscale=1.0,xscale=1.0]
	\draw[solid,arrows=-] (0,1ex) to (0,0) node[anchor=center,ground,rotate=-90,xshift=.66ex] {};
	\end{tikzpicture}
}

\begin{document}
\maketitle

\begin{abstract}
The ZX-calculus was introduced as a graphical language able to represent specific quantum primitives in an intuitive way. The recent completeness results have shown the theoretical possibility of a purely graphical description of quantum processes. However, in practice, such approaches are limited by the intrinsic low level nature of ZX calculus. The scalable notations have been proposed as an attempt to recover an higher level point of view while maintaining the topological rewriting rules of a graphical language. We demonstrate that the scalable ZX-calculus provides a formal, intuitive, and compact framework to describe and prove quantum algorithms. As a proof of concept, we consider the standard oracle-based quantum algorithms: Deutsch-Jozsa, Bernstein-Vazirani, Simon, and Grover algorithms, and we show they can be described and proved graphically.

\end{abstract}

\section*{Introduction}

The ZX-calculus is a graphical language for quantum computing \cite{coecke2011interacting}. This is a formal but also intuitive  language which captures fundamental properties of quantum mechanics. Contrary to the quantum circuit formalism, the ZX-calculus is equipped with a complete equational theory \cite{jeandel2018complete,HNW,jeandel2018diagrammatic}. This theoretical result makes the ZX-calculus a ideal tool for multiple applications in quantum computing. Among others we can cite the optimisation of quantum circuits \cite{kissinger2019reducing,duncan2019graph,de2020fast} and the design of fault tolerant quantum computation \cite{de2017zx,gidney2018efficient,hanks2020effective,Pauli-Fusion}. 
In this paper, we want to investigate one of the initial motivations of the ZX-calculus: make the ZX-calculus is a companion language for the description and the proof of quantum algorithms. We consider the standard oracle-based quantum algorithms: Deutsch-Jozsa \cite{deutsch1992rapid}, Bernstein-Vazirani \cite{bernstein1997quantum}, Simon \cite{simon1997power}, and Grover \cite{Gro97b} algorithms and  show they can be formulated and proved graphically. Our approach relies on the diagrammatic description of the quantum oracles and, in particular, a  graphical axiomatisation of the various promises the oracles satisfy.

A  section of the  "Dodo book" \cite{picturing-qp} is dedicated to the description of quantum algorithms in ZX-calculus (in particular Deutsch-Jozsa and Grover), and a few articles \cite{zeng2014abstract,vicary2013topological,gogioso2017fully} address the diagrammatic description of quantum oracles, but we show that the recent developments in the formalism, mainly the scalable construction \cite{carette2019szx,carette2020colored} and the discard construction \cite{carette2019completeness} allow for more self-contained, accurate and compact ZX-based proofs of quantum algorithms. Notice that depending on the algorithm we also use generators of the ZH-calculus \cite{backens2018zh} a variant of the ZX-calculus. 

%
%
%
%
%
%
%
%
%
\section{The ZX-calculus}

In this section we introduce various graphical tools. We use the ZX-calculus of \cite{coecke2011interacting} alongside the H harvestman of \cite{backens2018zh}. We work with the extension to mixed state quantum mechanics of \cite{carette2019completeness} and the scalable notations of \cite{carette2019szx}. Thus, the language could be denoted $S\text{ZXH}^{\ground}$-calculus, for simplicity we will rather say ZX-calculus. 

\subsection{String diagrams}

Our diagrams are arrows in a $C$-colored prop, i.e., a small strict symmetric monoidal category whose monoid of objects is spanned by a set $C$. We represent maps as colored string diagrams.  A map $f:a\to b$ is depicted as a box with input and output wires: $\input{f.tikz}$. By convention, we do not write the colors of the wires if this is clear from context. We do not draw wires of tensor unit type. We draw our diagrams from left to right. The composition $g\circ f: a\to c$ of two diagrams $f:a\to b$ and $g:b\to c$ is depicted by plugging diagrams:

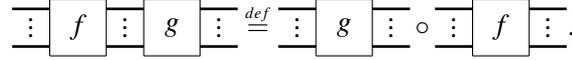
\begin{center}
	 $\input{cfg.tikz}\df\input{g.tikz}\circ\input{f.tikz}$.
\end{center}

The tensor $f\otimes g: a\otimes c\to b\otimes d$ of two diagrams $f:a\to b$ and $g:c\to d$ is depicted by juxtaposition of diagrams:

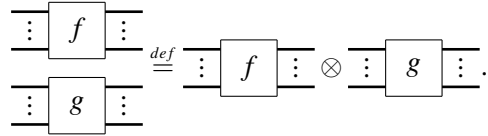
\begin{center}
	 $\input{tfg.tikz}\df\input{f.tikz}\otimes\input{g.tikz}$.
\end{center}

For any color $c\in C$, the identity $id_c:c\to c$ is depicted: $\begin{tikzpicture}
	\begin{pgfonlayer}{nodelayer}
		\node [style=none] (31) at (1, 0) {};
		\node [style=none] (32) at (0, 0) {};
	\end{pgfonlayer}
	\begin{pgfonlayer}{edgelayer}
		\draw [style=very thick] (31.center) to (32.center);
	\end{pgfonlayer}
\end{tikzpicture}
$. The swap map $\sigma_{c,c'}:c\otimes c'\to c'\otimes c$ is depicted: $\input{swap.tikz}$. $\sigma_{n,m}$ is a natural involution meaning:
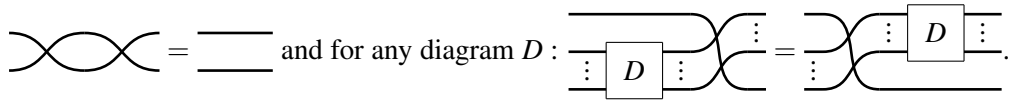
\begin{center}
	$\input{i0.tikz}=\input{i1.tikz}$ and for any diagram $D$ : $\input{nat0.tikz}=\input{nat1.tikz}$.
\end{center}

Our language also has a symmetric compact structure. For every color $c\in C$, there are two maps, cup map $\input{cup.tikz}: I \to c \otimes c$ and the cap map $\input{cap.tikz}: c\otimes c \to I$, satisfying:
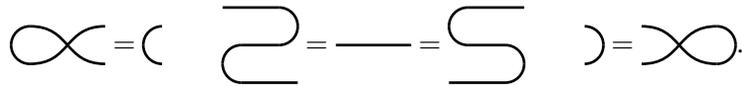
\begin{center}
	$\input{cs0.tikz}=\input{cs1.tikz}\qquad\input{s0.tikz}=\begin{tikzpicture}
	\begin{pgfonlayer}{nodelayer}
		\node [style=none] (32) at (2, 0) {};
		\node [style=none] (37) at (0, 0) {};
	\end{pgfonlayer}
	\begin{pgfonlayer}{edgelayer}
		\draw [style=very thick] (37.center) to (32.center);
	\end{pgfonlayer}
\end{tikzpicture}
=\input{s1.tikz}\qquad\input{cs3.tikz}=\input{cs2.tikz}$.
\end{center}

\subsection{Scalable notations}

Given a monochromatic prop $\mathcal{D}$ we define a colored prop $S\mathcal{D}$. $\mathcal{D}$ being monochromatic we denote $1$ the only color. All object are tensor products of copies of $1$, the tensor of $n$ copies is denoted $n$ and the tensor unit is denoted $0$. $S\mathcal{D}$ is a $\mathbb{N}^*$-colored prop, where $\mathbb{N}^*$ are the natural integers without zero. The colors are denoted $[n]$, the tensor unit is denoted $[0]$. The \textbf{size} of a type is inductively defined as: $|[n]|\df n$ and $|a\otimes b|\df|a|+ |b|$. We use thin wires to denote wires of type $[1]$ and thick wires denote any $[n]$. We write $[m]^n$ for the tensor product of $n$ wires $[m]$ with the convention $[m]^0=[0]$. For every generator $g:n\to m$ of $\mathcal{D}$ and every $k\in \mathbb{N}^*$ there is a generator $S_k(g):[k]^n \to [k]^m$ in $S\mathcal{D}$. Those scaled generators satisfy the same equations as the original ones. This induces a family of strict symmetric monoidal functors: $S_k:\mathcal{D}\to S\mathcal{D}$ defined on objects as $S_k(1)\df [k]$ and on morphisms as $S_k(g)\df S_k(g)$. $S\mathcal{D}$ contains also two fundamental families of generators, the dividers and the gatherers:

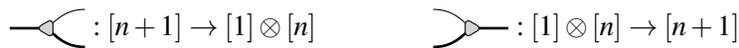
\begin{center}
	$\input{dd.tikz}: [n+1]\to [1]\otimes [n]\qquad\qquad\input{gg.tikz}: [1]\otimes [n]\to [n+1]$
\end{center}

satisfying:

\begin{center}
	$\input{ex0.tikz}=\begin{tikzpicture}
	\begin{pgfonlayer}{nodelayer}
		\node [style=none] (41) at (2, 0) {};
		\node [style=none] (42) at (0, 0) {};
	\end{pgfonlayer}
	\begin{pgfonlayer}{edgelayer}
		\draw [style=very thick] (42.center) to (41.center);
	\end{pgfonlayer}
\end{tikzpicture}
\qquad\qquad\input{el0.tikz}=
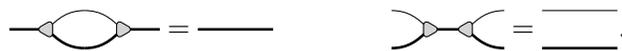
\begin{tikzpicture}
	\begin{pgfonlayer}{nodelayer}
		\node [style=none] (41) at (2, -0.5) {};
		\node [style=none] (42) at (0, -0.5) {};
		\node [style=none] (43) at (2, 0.5) {};
		\node [style=none] (44) at (0, 0.5) {};
	\end{pgfonlayer}
	\begin{pgfonlayer}{edgelayer}
		\draw [style=very thick] (42.center) to (41.center);
		\draw (44.center) to (43.center);
	\end{pgfonlayer}
\end{tikzpicture}
$.
\end{center}

From those equations we can deduce a coherence like result, the \textbf{rewiring lemma} proved in \cite{carette2020colored}:

\begin{lemma}[Rewiring]
	Given to types $a$ and $b$ such that $|a|=|b|$, there is a unique isomorphism $\gamma_{a,b}:a\to b$ made of dividers and gatherers.
\end{lemma}

In other words, any well typed equation involving only dividers and gatherers holds. The \textbf{boxing functor} $[\_]:S\mathcal{D}\to S\mathcal{D}$ is defined as $[a]\df [|a|]$ on objects and as $[f]\df \gamma_{b,[b]} \circ f\circ\gamma_{[a],a}$ on morphisms $f:a\to b$. Boxing is not a strict but a strong monoidal functor, in fact we have:

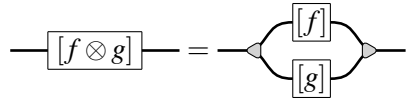
\begin{center}
	$\input{tbox0.tikz}=\input{tbox1.tikz}$
\end{center}

We define inductively the \textbf{thickening endofunctor}: $T_k: S\mathcal{D}\to S\mathcal{D}$ as $T_k([n])\df [kn]$,  $T_1\df id_{S\mathcal{D}}$ and 

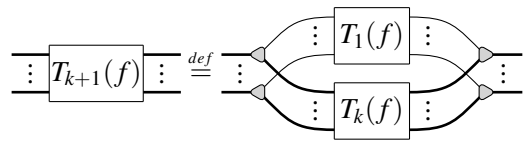
\begin{center}
	$\input{thick0.tikz}\df\input{thick1.tikz}$
\end{center}

Finally, for each generator $g$ in $\mathcal{D}$ we add the equations: $T_k(S_1(g))=S_k(g)$. Thus $T_l \circ S_k= S_{kl}$, this implies graphically:

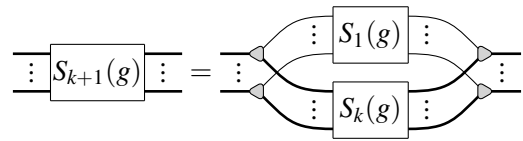
\begin{center}
	$\input{dist0.tikz}=\input{dist1.tikz}$
\end{center}

It is shown in \cite{carette2020colored} that any diagram in $S\mathcal{D}$ can be rewritten into normal form.

\begin{lemma}[Normal form]
	For each $f:a\to b$ in $S\mathcal{D}$ there is a unique $|f|:|a|\to |b|$ in $\mathcal{D}$ such that $f= \gamma_{[1]^{|b|},b} \circ S_1(|f|) \circ \gamma_{[1]^{|a|},a} $
\end{lemma}

This gives the \textbf{wire stripping} functor $|\_|:S\mathcal{D}\to \mathcal{D}$. Note that $f\mapsto |S_k(f)|$ is the multiplexing functor of \cite{chantawibul2018monoidal}. If $\mathcal{D}$ has a compact structure then so does $S\mathcal{D}$. In this situation, an iteration mechanism is available.

\begin{lemma}[Iteration]
Given any diagram $f:a\to a$ in $S\mathcal{D}$:

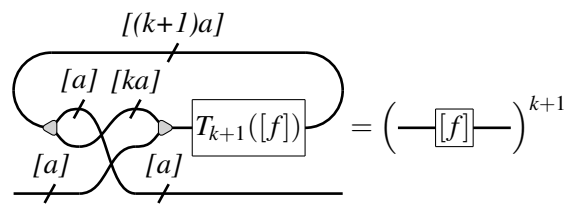
\begin{center}
	$\input{iter0.tikz}=\left(\input{iter1.tikz}\right)^{k+1}$
\end{center}
	
\end{lemma}

This allows to represent for loops graphically. In practice we need to compute the thickening of $[f]$. Thickening a scaled generator is by definition very easy, this only increases the size of every wires and generators. However, thickening a divider or a gatherer involves a permutation of the wires. This allows to apply the iteration mechanism as soon as we have a good representation of permutations.

\subsection{Arachnids}

Various graphical calculi have been introduced to represent quantum mechanics over qubits. The ZX-calculus was first introduced in \cite{coecke2011interacting}. Then the ZW-calculus in \cite{hadzihasanovic2018zw} and latter the ZH-calculus \cite{backens2018zh}. Those languages have their own specificities but are similar enough to be combined when needed. In this paper we will use a mix of ZX-calculus and ZH-calculus. The green and red families of spiders are indexed by phase vectors in $\left(\quot{\mathbb{R}}{2\pi\mathbb{Z}}\right)^k$, by convention the phase is $0$ if not given. The yellow family of harvestmen (a variation of spiders obeying a modified fusion rule) is indexed by complex vectors, by convention the phase is $-1$ if not given. They are respectively depicted:

\begin{center}
	$\input{z.tikz}:\cdot[k]^n\to [k]^m\qquad\input{x.tikz}:[k]^n\to [k]^m\qquad\input{h.tikz}:[k]^n\to [k]^m$
\end{center}

Denoting $\alpha_1$ {(resp. $x_1$)} the head of the phase vector $\alpha$ (resp. $x$) and $\alpha'$ (resp. $x'$) its tail, the arachnids interact with dividers and gatherers with:

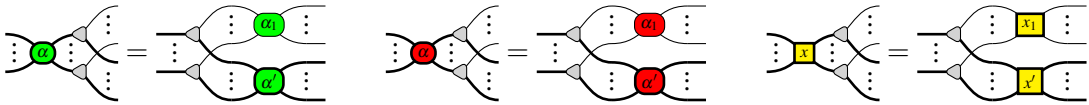
\begin{center}
	$\input{dig0.tikz}=\input{dig1.tikz}\qquad\input{dir0.tikz}=\input{dir1.tikz}\qquad\input{dih0.tikz}=\input{dih1.tikz}$
\end{center}

All of them are flexsymmetric that is for every arachnid $g:a\to b$ and every permutation of wires $\sigma$:

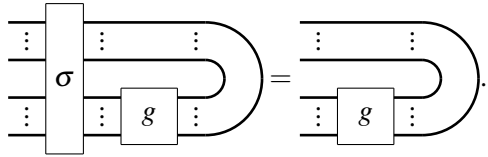
\begin{center}
	$\input{fl0.tikz}=\input{fl1.tikz}$.
\end{center}  

Furthermore, we can bend the legs of the arachnids: 

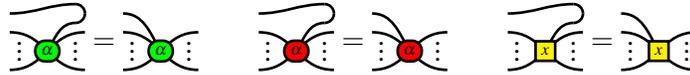
\begin{center}
	$\input{dg0.tikz}=\input{dg1.tikz}\qquad\input{dr0.tikz}=\input{dr1.tikz}\qquad\input{dh0.tikz}=\input{dh1.tikz}$
\end{center}

They all satisfy fusion rules:

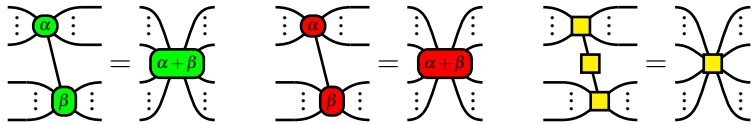
\begin{center}
	$\input{gar0.tikz}=\input{gar1.tikz}\qquad\input{rar0.tikz}=\input{rar1.tikz}\qquad\input{hha0.tikz}=\input{hha1.tikz}$
\end{center}

Note we only define fusion of yellow boxes indexed by phase $-1$.

The three arachnids interact with each other in the following way:

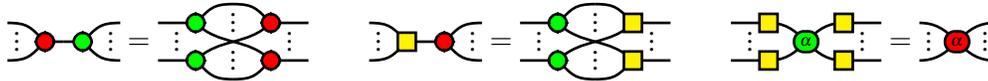
\begin{center}
	$\input{rgb0.tikz}=\input{rgb1.tikz}\qquad\input{hgb0.tikz}=\input{hgb1.tikz}\qquad\input{hgr0.tikz}=\input{hgr1.tikz}$
\end{center}

\subsection{Quantum gates}

Quantum processes are maps in the category $\textbf{CPM}_2$ whose objects are the sets $\mathcal{M}_{2^n\times2^n}(\mathbb{C})$ of matrices  and arrows are the completely positive maps. $\mathcal{M}_{2\times2}(\mathbb{C})$ corresponds to a qubit and $\mathcal{M}_{2^n\times2^n}(\mathbb{C})$ to a register of $n$ qubits. We denote $\textbf{2}$ the set $\{0,1\}$. $\ket{x}$ with $x\in \textbf{2}^n$ denotes the canonical basis of $\mathbb{C}^{2^n}$. We write $\bra{x}\df \ket{x}^\dagger$ where $\dagger$ is the Hermitian adjoint. To each type we associate an Hilbert space and to each diagram we associate a completely positive map. In other words, there is an \textbf{interpretation functor} $\interp{\_}:S\mathcal{D}\to \textbf{CPM}_2$. We have $\interp{[n]}\df\mathcal{M}_{2^n\times2^n}(\mathbb{C})$ and $\interp{a\otimes b}\df \interp{a}\otimes\interp{b}$. Note that in general $[a+b]\neq [a]\otimes [b]$, but we always have $\interp{[a+b]}\simeq \interp{[a]\otimes [b]}\simeq \mathcal{M}_{2^{a+b}\times2^{a+b}}(\mathbb{C})$. The wires have interpretations.

\begin{center}
	$\interp{\begin{tikzpicture}
	\begin{pgfonlayer}{nodelayer}
		\node [style=none] (31) at (1, 0) {};
		\node [style=none] (32) at (0, 0) {};
	\end{pgfonlayer}
	\begin{pgfonlayer}{edgelayer}
		\draw [style=very thick] (31.center) to (32.center);
	\end{pgfonlayer}
\end{tikzpicture}
}\df \rho \mapsto \rho \quad\interp{\input{swap.tikz}}\df \rho \otimes \rho' \mapsto \rho' \otimes \rho\quad\interp{\input{cup.tikz}}\df \sum\limits_{x\in \textbf{2}^n}\ket{xx}\bra{xx}\quad\interp{\input{cap.tikz}}\df \rho\mapsto \sum\limits_{x\in \textbf{2}^n}\bra{xx}\rho\ket{xx}$.
\end{center}

The dividers and gatherers act trivially:
 
\begin{center}
	$\interp{\input{dd.tikz}}\df \rho\mapsto \rho \qquad\interp{\input{gg.tikz}}\df \rho\mapsto \rho$
\end{center}

We use the well tempered normalization of \cite{de2020well} setting $\interp{\bigstar}\df \frac{1}{\sqrt{2}}$.  
The arachnids of type $[k]^n\to [k]^m$ have interpretations:

\begin{center}
	\begin{tabular}{l}
		$\interp{\input{z.tikz}}\df \rho \mapsto V\rho V^\dagger $ with  $V\df 2^{k\frac{n+m-2}{4}}\sum\limits_{x\in {\bf 2}^k} e^{i(x\cdot a)}\ket{x}^{\otimes n}\!\bra{x}^{\otimes m}$\\
		\\
		$\interp{\input{x.tikz}}\df \rho \mapsto V\rho V^\dagger $ with  $V\df 2^{k\frac{2-n-m}{4}}\sum\limits_{x_i \in {\bf 2}^k} \prod\limits_{j=1}^{k}\frac{1+e^{i \left(a_j + \sum\limits_{i=1}^{n+m}x_{i,j}\right)}}{2}\ket{x_1 \cdots x_n}\!\bra{x_{n+1} \cdots x_{n+m}}$\\
		\\
		$\interp{\input{h.tikz}}\df \rho \mapsto V\rho V^\dagger $ with  $V\df 2^{-k\frac{n+m}{4}}\sum\limits_{x_i\in {\bf 2}^k} \prod\limits_{j=1}^{k} x^{\bigwedge\limits_{i=1}^{n+m} x_{i,j}}\ket{x_1 \cdots x_n}\!\bra{x_{n+1} \cdots x_{n+m}}$
	\end{tabular}
\end{center}
Where the $x_i$ are binary words of size $k$ and $x_{i,j}$ is the $j$-th bit of $x_i$. By convention if the phase is not given it is $0$ for red and green spiders and -$1$ for yellow ones. Usually quantum algorithms are presented as quantum circuits built from elementary quantum gates. Our language is expressive enough to represent all of them. The main idea is that our generators decompose quantum gates into more fundamental parts which equational theory is better understood. The most common states are represented:

\setlength{\tabcolsep}{7pt}
\renewcommand{\arraystretch}{1}

\begin{center}
	\begin{tabular}{|l|c|c|c|c|}
		\hline
		State & $\ket{0}$ & $\ket{1}$ & $\frac{\ket0+\ket 1}{\sqrt 2}$& $\frac{\ket0-\ket 1}{\sqrt 2}$ \\
		\hline
		Diagram &$\begin{tikzpicture}
	\begin{pgfonlayer}{nodelayer}
		\node [style=rdot] (51) at (0.5, 0) {};
		\node [style=none] (52) at (1.5, 0) {};
		\node [style=none] (53) at (0, 0.5) {$\bigstar$};
	\end{pgfonlayer}
	\begin{pgfonlayer}{edgelayer}
		\draw (51) to (52.center);
	\end{pgfonlayer}
\end{tikzpicture}
$&$\begin{tikzpicture}
	\begin{pgfonlayer}{nodelayer}
		\node [style=rdot] (51) at (0.5, 0) {$\pi$};
		\node [style=none] (52) at (1.5, 0) {};
		\node [style=none] (53) at (0, 0.5) {$\bigstar$};
	\end{pgfonlayer}
	\begin{pgfonlayer}{edgelayer}
		\draw (51) to (52.center);
	\end{pgfonlayer}
\end{tikzpicture}
$&$\begin{tikzpicture}
	\begin{pgfonlayer}{nodelayer}
		\node [style=gdot] (51) at (0.5, 0) {};
		\node [style=none] (52) at (1.5, 0) {};
		\node [style=none] (53) at (0, 0.5) {$\bigstar$};
	\end{pgfonlayer}
	\begin{pgfonlayer}{edgelayer}
		\draw (51) to (52.center);
	\end{pgfonlayer}
\end{tikzpicture}
$&$\begin{tikzpicture}
	\begin{pgfonlayer}{nodelayer}
		\node [style=gdot] (51) at (0.5, 0) {$\pi$};
		\node [style=none] (52) at (1.5, 0) {};
		\node [style=none] (53) at (0, 0.5) {$\bigstar$};
	\end{pgfonlayer}
	\begin{pgfonlayer}{edgelayer}
		\draw (51) to (52.center);
	\end{pgfonlayer}
\end{tikzpicture}
$\\
		\hline
		Density Matrix &\footnotesize$\begin{pmatrix}
		1&0\\ 0&0
		\end{pmatrix}$&\footnotesize$\begin{pmatrix}
		0&0\\
		0&1
		\end{pmatrix}$&\footnotesize$\begin{pmatrix}
		\frac{1}{2}&\frac{1}{2}\\
		\frac{1}{2}&\frac{1}{2}
		\end{pmatrix}$&\footnotesize$\begin{pmatrix}
		\frac{1}{2}&-\frac{1}{2}\\
		-\frac{1}{2}&\frac{1}{2}
		\end{pmatrix}$\\
		\hline
	\end{tabular}
\end{center}

We will also use the mirror image of those states corresponding to effects. By doing so we will obtain post-selected circuits and we will be able to compute amplitudes.

The following table provides the representation of the most common quantum gates gates. They are all pure maps, i.e., operators of the form: $\rho \mapsto V\rho V^\dagger$. We just give the corresponding matrix $V$.

\begin{center}
	\begin{tabular}{|c|c|c|c|c|c|c|c|}
		\hline
		Name & H & Not & Z & Swap & C-Not & C-Z & Toffoli\\ 
		\hline
		Gate &$\input{qhgt.tikz}$&$\input{qngt.tikz}$&$\input{qzgt.tikz}$&$\input{qswapgt.tikz}$&$\input{qcngt.tikz}$&$\input{qczgt.tikz}$&$\input{qtofgt.tikz}$\\ 
		\hline
		Diagram &$\begin{tikzpicture}
	\begin{pgfonlayer}{nodelayer}
		\node [style=none] (0) at (0, 0) {};
		\node [style=none] (1) at (1, 0) {};
		\node [style=none] (2) at (2, 0) {};
		\node [style=had] (3) at (1, 0) {};
	\end{pgfonlayer}
	\begin{pgfonlayer}{edgelayer}
		\draw (0.center) to (1.center);
		\draw (1.center) to (2.center);
	\end{pgfonlayer}
\end{tikzpicture}
$&$\begin{tikzpicture}
	\begin{pgfonlayer}{nodelayer}
		\node [style=none] (0) at (0, 0) {};
		\node [style=rdot] (1) at (1, 0) {$\pi$};
		\node [style=none] (2) at (2, 0) {};
	\end{pgfonlayer}
	\begin{pgfonlayer}{edgelayer}
		\draw (0.center) to (1);
		\draw (1) to (2.center);
	\end{pgfonlayer}
\end{tikzpicture}
$&$\begin{tikzpicture}
	\begin{pgfonlayer}{nodelayer}
		\node [style=none] (0) at (0, 0) {};
		\node [style=gdot] (1) at (1, 0) {$\pi$};
		\node [style=none] (2) at (2, 0) {};
	\end{pgfonlayer}
	\begin{pgfonlayer}{edgelayer}
		\draw (0.center) to (1);
		\draw (1) to (2.center);
	\end{pgfonlayer}
\end{tikzpicture}
$&$\input{swapgt.tikz}$&$\input{cngt.tikz}$&$\input{czgt.tikz}$&$\input{tofgt.tikz}$\\
		\hline
		$V$ &\footnotesize$\begin{pmatrix}
		\frac{1}{\sqrt{2}}&\frac{1}{\sqrt{2}}\\
		\frac{1}{\sqrt{2}}&\frac{-1}{\sqrt{2}}
		\end{pmatrix}$&\footnotesize$\begin{pmatrix}
		0&1\\
		1&0
		\end{pmatrix}$&\footnotesize$\begin{pmatrix}
		1&0\\
		0&-1
		\end{pmatrix}$&\scalebox{0.5}{$\begin{pmatrix}
			1&0&0&0\\
			0&0&1&0\\
			0&1&0&0\\
			0&0&0&1
			\end{pmatrix}$}&\scalebox{0.5}{$\begin{pmatrix}
			1&0&0&0\\
			0&1&0&0\\
			0&0&0&1\\
			0&0&1&0
			\end{pmatrix}$}&\scalebox{0.5}{$\begin{pmatrix}
			1&0&0&0\\
			0&1&0&0\\
			0&0&1&0\\
			0&0&0&-1
			\end{pmatrix}$}&\scalebox{0.5}{$\begin{pmatrix}
			1&0&0&0&0&0&0&0\\
			0&1&0&0&0&0&0&0\\
			0&0&1&0&0&0&0&0\\
			0&0&0&1&0&0&0&0\\
			0&0&0&0&1&0&0&0\\
			0&0&0&0&0&1&0&0\\
			0&0&0&0&0&0&0&1\\
			0&0&0&0&0&0&1&0
			\end{pmatrix}$} \\
		\hline
	\end{tabular}
\end{center}

We will also use the \textbf{discard map} and its transpose which is the non normalized \textbf{completly mixed state}:

\begin{center}
	$\interp{\input{mix.tikz}}\df \rho\mapsto Tr(\rho)\qquad\interp{\input{disc.tikz}}\df \rho\mapsto \sum\limits_{x\in \textbf{2}} \ket{x}\bra{x}$
\end{center}

The first corresponds to discarding data and the second is a uniform probabilistic mixture of states. The isometries, the maps such that $\input{isom0.tikz}=\begin{tikzpicture}
	\begin{pgfonlayer}{nodelayer}
		\node [style=none] (48) at (0, 0) {};
		\node [style=none] (49) at (2, 0) {};
	\end{pgfonlayer}
	\begin{pgfonlayer}{edgelayer}
		\draw [style=very thick] (48.center) to (49.center);
	\end{pgfonlayer}
\end{tikzpicture}
$, satisfies: $\input{causal0.tikz}=\input{causal1.tikz}$.


\section{The calculus of oracles}

Many quantum algorithms are defined using oracles. An oracle can be viewed as a black box, it is not however an arbitrary map, a quantum oracle may have some structure: they are usually quantum encodings of classical functions,  moreover some promises can provide additional informations about the behaviour of the oracle. In the spirit of the ZX-calculus, we decompose, in this section, classes of quantum oracles into smaller components with better understood algebraic properties.

\subsection{Function arrows}

Given a function $f:\textbf{2}^n \to \textbf{2}^m$ we define a \textbf{function arrow}: $\input{farrow.tikz}:[n]\to [m]$.

\begin{center}
	$\interp{\input{farrow.tikz}}=\rho \mapsto V\rho V^\dagger $ with $V\df\ket{x}\mapsto 2^{\frac{m-n}{4}}\ket{f(x)}$
\end{center}

Any function arrow satisfies:

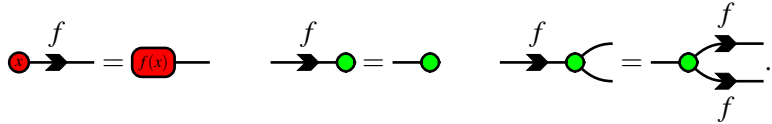
\begin{center}
	 $\input{fapply0.tikz}=\begin{tikzpicture}
	\begin{pgfonlayer}{nodelayer}
		\node [style=grdot] (66) at (0, 0) {$f(x)$};
		\node [style=none] (71) at (1.5, 0) {};
	\end{pgfonlayer}
	\begin{pgfonlayer}{edgelayer}
		\draw [style=very thick] (66) to (71.center);
	\end{pgfonlayer}
\end{tikzpicture}
\qquad\input{fgerase0.tikz}=\begin{tikzpicture}
	\begin{pgfonlayer}{nodelayer}
		\node [style=ggdot] (69) at (1, 0) {};
		\node [style=none] (71) at (0, 0) {};
	\end{pgfonlayer}
	\begin{pgfonlayer}{edgelayer}
		\draw [style=very thick] (71.center) to (69);
	\end{pgfonlayer}
\end{tikzpicture}
\qquad\input{fgcopy0.tikz}=\input{fgcopy1.tikz}$.
\end{center}
Where, with a slight abuse of notation, the $\pi$ factor is omitted, e.g. $\begin{tikzpicture}
	\begin{pgfonlayer}{nodelayer}
		\node [style=grdot] (51) at (0.5, 0) {$x\pi$};
		\node [style=none] (52) at (1.5, 0) {};
	\end{pgfonlayer}
	\begin{pgfonlayer}{edgelayer}
		\draw[style=very thick]  (51) to (52.center);
	\end{pgfonlayer}
\end{tikzpicture}
$ is simply depicted $\begin{tikzpicture}
	\begin{pgfonlayer}{nodelayer}
		\node [style=grdot] (51) at (0.5, 0) {$x$};
		\node [style=none] (52) at (1.5, 0) {};
	\end{pgfonlayer}
	\begin{pgfonlayer}{edgelayer}
		\draw[style=very thick]  (51) to (52.center);
	\end{pgfonlayer}
\end{tikzpicture}
$, when it is clear that $x$ is a binary vector. 
We can visualize some properties of function as graphical equations for the corresponding function arrow:

\begin{lemma}
	Given a function arrow $f$:
	
	\begin{center}
		\begin{tabular}{c}
			$f$ is \textbf{balanced} $\quad\stackrel{def}{\Leftrightarrow}\quad\forall x,y\in 2^m ~ |f^{-1}\left(\{x\}\right)|=|f^{-1}\left(\{y\}\right)|\quad\Leftrightarrow\quad\input{fgcoerase0.tikz}=\begin{tikzpicture}
	\begin{pgfonlayer}{nodelayer}
		\node [style=ggdot] (69) at (0, 0) {};
		\node [style=none] (71) at (1, 0) {};
	\end{pgfonlayer}
	\begin{pgfonlayer}{edgelayer}
		\draw [style=very thick] (71.center) to (69);
	\end{pgfonlayer}
\end{tikzpicture}
$ \\
			$f$ is \textbf{injective} $\quad\stackrel{def}{\Leftrightarrow}\quad\forall x,y\in 2^n ~ \left(f(x)=f(y) \Rightarrow x=y\right)\quad\Leftrightarrow\quad\input{fgcocopy0.tikz}=\input{fgcocopy1.tikz}$
		\end{tabular}
	\end{center}
	
\end{lemma}

We see that being balanced implies being surjective. But the converse is not true.

\subsection{Red arrows}

A function $f:{\bf 2}^n \to {\bf 2}^m$ can be seen as a map $f:\mathbb{F}_2^n \to \mathbb{F}_2^m$, if this map is $\mathbb{F}_2$-linear then it can be described by a matrix $A\in \mathcal{M}_{m\times n}\left(\mathbb{F}_2\right)$. The \textbf{red matrix arrows} indexed by $A$ is defined by $\input{rarrow.tikz}\df \input{farrow.tikz}$. By convention if the index matrix is not given it is assumed that it is the full of one matrix. Those arrows have been extensively studied in \cite{carette2019szx} and \cite{carette2020colored}. We recall the main properties of red matrix arrows. First the $\mathbb{F}_2$-linearity translates to: $\input{Arcoerase0.tikz}=
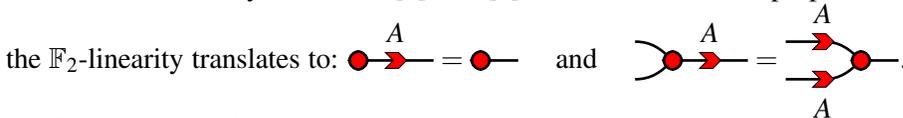
\begin{tikzpicture}
	\begin{pgfonlayer}{nodelayer}
		\node [style=grdot] (69) at (0, 0) {};
		\node [style=none] (71) at (1, 0) {};
	\end{pgfonlayer}
	\begin{pgfonlayer}{edgelayer}
		\draw [style=very thick] (71.center) to (69);
	\end{pgfonlayer}
\end{tikzpicture}
\quad$ and $\quad\input{Arcocopy0.tikz}=\input{Arcocopy1.tikz}$.

They interact with the dividers and gatherers as:

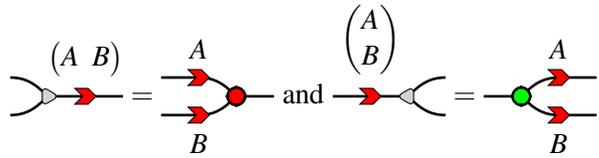
\begin{center}
	$\input{rowred0.tikz}=\input{rowred1.tikz}$ and $\input{colred0.tikz}=\input{colred1.tikz}$.
\end{center}

Most of the properties of red matrix arrows can be summed up into one meta rule:

\begin{center}
	$\input{span0.tikz}=\input{span1.tikz} \Leftrightarrow \quad Im\begin{pmatrix}
	C\\D
	\end{pmatrix}=Ker\begin{pmatrix}
	A&B
	\end{pmatrix}$
\end{center}

With $k\df dim\left(Ker\begin{pmatrix}
C\\D
\end{pmatrix}\right)$ and $h\df dim\left(coKer\begin{pmatrix}
A&B
\end{pmatrix}\right)$.

A last important rule is the interaction with the Hadamard gate: $\input{ihad0.tikz}=\input{ihad1.tikz}$.


Red matrix arrows were the first motivation to the definition of scalable notations in \cite{chancellor2016graphical} where they are applied to the design of error correcting codes. 


\subsection{Yellow arrows}

As noted in \cite{carette2020colored}, the possibility to index arrows by matrices is linked to a bi-algebra structure. If the red/green bi-algebra leads to red matrix arrows, the yellow/green bi-algebra gives us another family of matrix arrows over the boolean semi-ring $\mathbb{B}\df \left(\{0,1\}, \land, \lor\right)$. A function $f:{\bf 2}^n \to {\bf 2}^m$ can be seen as a map $f:\mathbb{B}^n \to \mathbb{B}^m$, if this map is a homomorphism of $\mathbb{B}$-semi module, that is $f(a\land b)=f(a)\land f(b)$ and $f(1)=1$, then it can be described by a matrix $A\in \mathcal{M}_{m\times n}\left(\mathbb{B}\right)$. The \textbf{yellow matrix arrow} indexed by $A$ is then defined by $\input{yarrow.tikz}\df \input{farrow.tikz}$. We take the same convention as red arrows, if no matrix is given then it is imply that the arrow is indexed by the full of one matrix. Being a homomorphism of $\mathbb{B}$-semi module translates to:	$\input{yArcoerase0.tikz}=
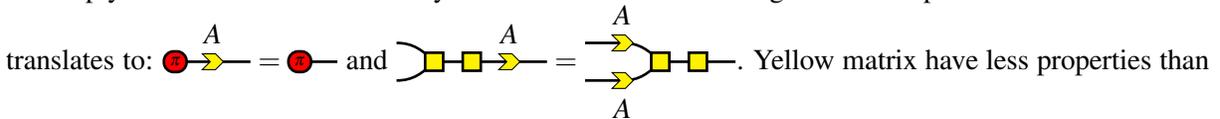
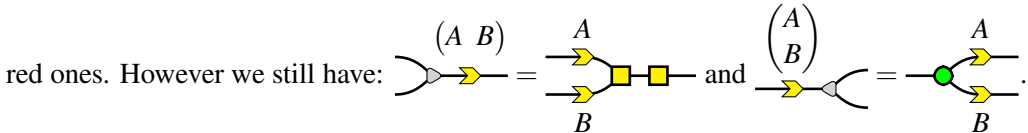
\begin{tikzpicture}
	\begin{pgfonlayer}{nodelayer}
		\node [style=grdot] (69) at (0, 0) {$\pi$};
		\node [style=none] (71) at (1, 0) {};
	\end{pgfonlayer}
	\begin{pgfonlayer}{edgelayer}
		\draw [style=very thick] (71.center) to (69);
	\end{pgfonlayer}
\end{tikzpicture}
$ and $\input{yArcocopy0.tikz}=\input{yArcocopy1.tikz}$. Yellow matrix have less properties than red ones. However we still have: $\input{rowyel0.tikz}=\input{rowyel1.tikz}$ and $\input{colyel0.tikz}=\input{colyel1.tikz}$.

\subsection{Quantum oracle}

The function arrow of $f$ is unitary if and only if $f$ is a bijection. There is however a standard way to associate with any function $f:{\bf 2}^n \to {\bf 2}^m$ a unitary transformation defined as $U_f=\ket{x}\ket{y}\mapsto \ket{x}\ket{f(x)\oplus y}$, often call \textbf{quantum oracle}. As pointed out in \cite{picturing-qp}, the quantum oracle can be constructed as follows:  $\input{foracle.tikz}: [n]\otimes [m] \to [n]\otimes [m]$. Indeed, $\interp{\input{foracle.tikz}}= \ket{x}\ket{y}\mapsto \ket{x}\ket{f(x)\oplus y}$. We can double check that quantum oracles are  involutions:

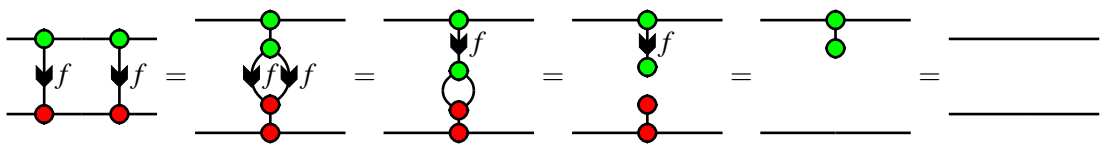
\begin{center}
	$\input{finv0.tikz}=\input{finv1.tikz}=\input{finv2.tikz}=\input{finv3.tikz}=\input{finv5.tikz}=\input{finv4.tikz}$
\end{center}

From a quantum oracle we can easily compute the original function using ancillas.

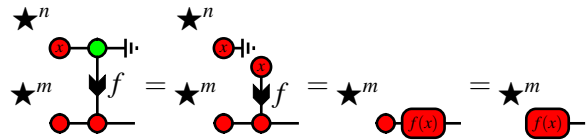
\begin{center}
	$\input{fc0.tikz}=\input{fc1.tikz}=\input{fc2.tikz}=\begin{tikzpicture}
	\begin{pgfonlayer}{nodelayer}
		\node [style=none] (34) at (0.75, -1) {};
		\node [style=none] (51) at (1.75, -1) {};
		\node [style=grdot] (53) at (0.75, -1) {$f(x)$};
		\node [style=none] (65) at (0, -0.25) {$\bigstar^m$};
	\end{pgfonlayer}
	\begin{pgfonlayer}{edgelayer}
		\draw [style=very thick] (34.center) to (51.center);
	\end{pgfonlayer}
\end{tikzpicture}
$
\end{center}

Note that the Toffoli gate is in fact the quantum oracle representing the $AND$ gate. Often, we consider boolean functions, then, another kind of oracle is available. For any boolean function $f:\textbf{2}^n\to \textbf{2}$, the diagonal oracle of $f$ is $\input{fdoracle.tikz}: [n] \to [n]$: $\interp{\input{fdoracle.tikz}}= \ket{x}\mapsto (-1)^{f(x)}\ket{x}$. We can construct the diagonal oracle from the oracle using ancillas:

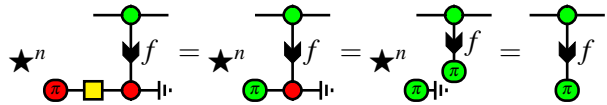
\begin{center}
	$\input{fd0.tikz}=\input{fd1.tikz}=\input{fd2.tikz}=\input{fd3.tikz}$
\end{center}

We have now enough graphical structures to tackle the most basic quantum algorithms.

\section{Quantum Algorithms}

In this section we provide a diagrammatic treatment of some quantum algorithms that frequently appear in quantum computing textbooks  like \cite{nielsen2002quantum}. They are oracle-based: given a function which satisfies some properties (the promise), we want to recover some information about the function using  a minimal number of queries to the corresponding quantum oracle.

\subsection{Bernstein-Vazirani}

The Bernstein-Vazirani algorithm has been introduced in \cite{bernstein1997quantum}. The goal is to recover a string of bits encoded into a function. 

\begin{center}
	\begin{tabular}{|ll|}
		\hline
		&\\
		\textbf{Input:}& A function $f:\{0,1\}^n \to \{0,1\}$ of the form $f(x)=s^t\cdot x$ with $s\in \{0,1\}^n $.\\
		&\\
		\textbf{Problem:}& Find $s$.\\
		&\\
		\textbf{Circuit:}& $\input{DJprime.tikz} \quad\to\quad \input{DJ01.tikz}$\\
		&\\
		\hline
	\end{tabular}
\end{center}

Reformulating graphically the promise on $f$ gives us: $\input{farrow0.tikz}=\input{sarrow.tikz}$ and then:

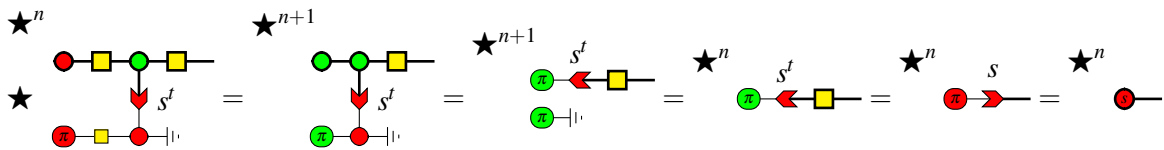
\begin{center}
	$\input{BV5.tikz}=\input{BV6.tikz}=\input{BV7.tikz}=\input{BV8.tikz}=\input{BV9.tikz}=\begin{tikzpicture}
	\begin{pgfonlayer}{nodelayer}
		\node [style=grdot] (60) at (1, 0) {$s$};
		\node [style=none] (65) at (2, 0) {};
		\node [style=none] (70) at (0, 1) {$\bigstar^{n}$};
	\end{pgfonlayer}
	\begin{pgfonlayer}{edgelayer}
		\draw [style=very thick] (60) to (65.center);
	\end{pgfonlayer}
\end{tikzpicture}
$
\end{center}

We see the circuit directly outputs the state $\ket{s}$.

\subsection{Deutsch-Jozsa}

The Deutsch-Jozsa algorithm \cite{deutsch1992rapid} is historically the first of all quantum algorithms. 
Given a function that is known to be either constant or balanced, the goal is to decide in which case we are using only one query to the oracle. The version we present here is a little bit more general than usual, since we do not require $f$ to output a single bit. The general principle is the same as Bernstein-Vazirani,
the difference is that we are here interested in the probability of outputting $\ket{0}^{\otimes n}$.

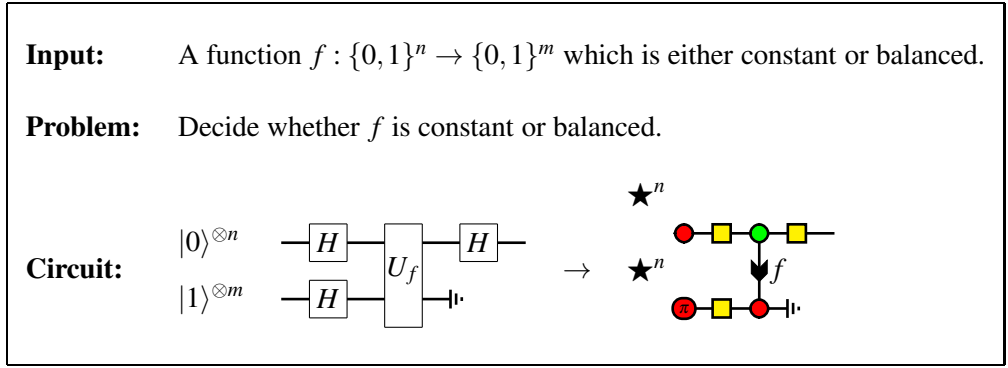
\begin{center}
	\begin{tabular}{|ll|}
		\hline
		&\\
		\textbf{Input:}& A function $f:\{0,1\}^n \to \{0,1\}^m$ which is either constant or balanced.\\
		&\\
		\textbf{Problem:}& Decide whether $f$  is constant or balanced.\\
		&\\
		\textbf{Circuit:}& $\input{DJ.tikz} \quad\to\quad \input{DJ0.tikz}$\\
		&\\
		\hline
	\end{tabular}
\end{center}

We compute the amplitude of the outcome $\ket{0}^{\otimes n}$:

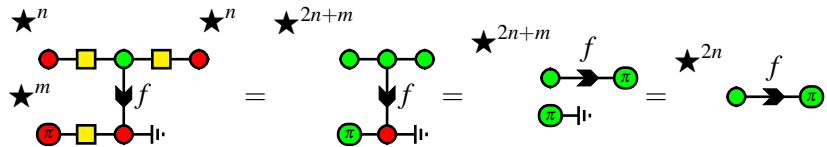
\begin{center}
	$\input{DJ00.tikz}=\input{DJ1.tikz}=\input{DJ2.tikz}=\input{DJ3.tikz}$
\end{center}

We then have two cases:

\begin{itemize}
	\item[$\bullet$] If $f$ is balanced then $\input{fgcoerase0.tikz}=\begin{tikzpicture}
	\begin{pgfonlayer}{nodelayer}
		\node [style=ggdot] (69) at (0, 0) {};
		\node [style=none] (71) at (1, 0) {};
	\end{pgfonlayer}
	\begin{pgfonlayer}{edgelayer}
		\draw [style=very thick] (71.center) to (69);
	\end{pgfonlayer}
\end{tikzpicture}
$ and $\input{DJ3.tikz}=\begin{tikzpicture}
	\begin{pgfonlayer}{nodelayer}
		\node [style=ggdot] (57) at (0, 0) {$\pi$};
		\node [style=ggdot] (60) at (1, 0) {};
		\node [style=none] (90) at (2.05, 1) {$\bigstar^{2n}$};
	\end{pgfonlayer}
	\begin{pgfonlayer}{edgelayer}
		\draw [style=very thick] (57) to (60);
	\end{pgfonlayer}
\end{tikzpicture}
=0$.
	
	\item[$\bullet$] If $f$ is constant then there exists $x\in \{0,1\}^m$ such that $\input{farrow.tikz}=\begin{tikzpicture}
	\begin{pgfonlayer}{nodelayer}
		\node [style=ggdot] (0) at (1, 0) {};
		\node [style=none] (7) at (0, 0) {};
		\node [style=grdot] (8) at (2, 0) {$x$};
		\node [style=none] (9) at (3, 0) {};
	\end{pgfonlayer}
	\begin{pgfonlayer}{edgelayer}
		\draw [style=very thick] (7.center) to (0);
		\draw [style=very thick] (8) to (9.center);
	\end{pgfonlayer}
\end{tikzpicture}
$:
	
	\begin{center}
		$\input{DJ3.tikz}=\input{DJ5.tikz}=\begin{tikzpicture}
	\begin{pgfonlayer}{nodelayer}
		\node [style=ggdot] (57) at (0, 0) {$\pi$};
		\node [style=grdot] (91) at (1, 0) {$x$};
	\end{pgfonlayer}
	\begin{pgfonlayer}{edgelayer}
		\draw [style=very thick] (57) to (91);
	\end{pgfonlayer}
\end{tikzpicture}
=1$.
	\end{center}
	
\end{itemize}

So if the outcome is $\ket{0}^{\otimes n}$ then $f$ is constant otherwise $f$ is balanced.

\subsection{Simon}

Simon's algorithm 
is more subtle than the algorithm we have seen so far. This algorithm is probabilistic, moreover the quantum computation is combined with a classical processing. 
We are given a strictly periodic function $f$ and the goal is to find the period $s$. The quantum part of the algorithm is nothing but a random generator that outputs a string $y$ such that $y\cdot s=0$, in a uniform way. Repeating this quantum part several times, this gives, with high probability, enough linearly independent equations to solve the linear system with a classical algorithm and find $s$.

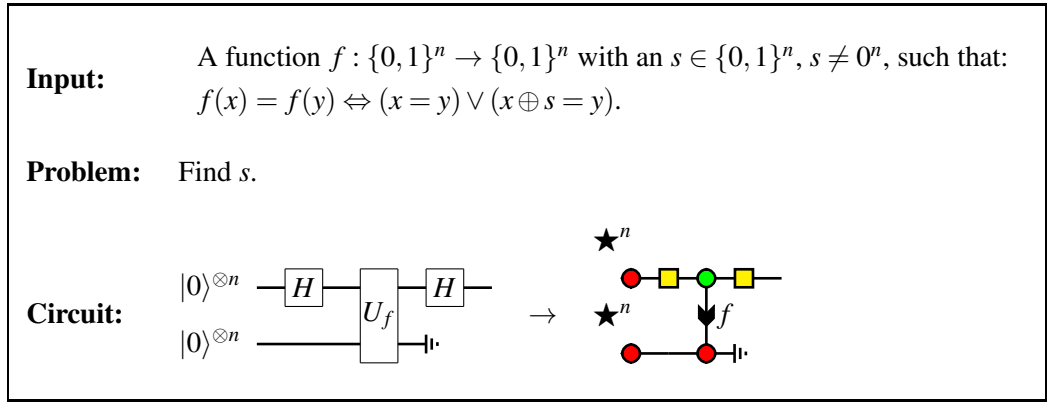
\begin{center}
	\begin{tabular}{|ll|}
		\hline
		&\\
		\textbf{Input:}& \begin{tabular}{l}
		A function $f:\{0,1\}^n \to \{0,1\}^{n}$ with an $s\in \{0,1\}^n $, $s\neq 0^n$, such that:\\[0.1cm]
		\multicolumn{1}{l}{$f(x)=f(y) \Leftrightarrow  (x=y) \lor (x\oplus s=y) $.}
		\end{tabular}\\
		&\\
		\textbf{Problem:}& Find $s$.\\
		&\\
		\textbf{Circuit:}& $\input{sim0.tikz} \quad\to\quad \input{sim1.tikz}$\\
		&\\
		\hline
	\end{tabular}
\end{center}
	
The translation of the promise into a graphical property is less straightforward than with the algorithms we have seen so far. Let $h$ be an orthogonal projector on $s^\perp$, $h$ is clearly strictly $s$ periodic. So there is a bijective function $g:\{0,1\}^{n} \to \{0,1\}^{n}$ such that $\input{farrow.tikz}=\input{simarrow.tikz}$. The circuit reduces to: 
	
	\begin{center}
		$\input{sim01.tikz}=\input{sim2.tikz}=\input{sim3.tikz}=\input{sim5.tikz}=\input{sim6.tikz}=\input{sim7.tikz}$.
	\end{center}
	
Since by definition $h^t=h$. We can directly see the resulting state: it is a uniform mixture of the elements in $s^\perp$. In other words, we can use this circuit to sample uniformly at random vectors $y_i$ such that $y_i\cdot s=0$.

\subsection{Grover}

The last and most famous algorithm we present is the Grover's Algorithm \cite{Gro97b}. Given a boolean function $f:{\bf 2}^n\to {\bf 2}$ such that $1$ has a unique preimage $x$. The objective is to find $x$. 
Roughly speaking, the algorithm consists in applying $k$ times a combination of the quantum oracle and a diffusion operator 
on the superposition of all classical inputs. We then show that choosing $k$ wisely, the output is $\ket{x}$ with a high probability.

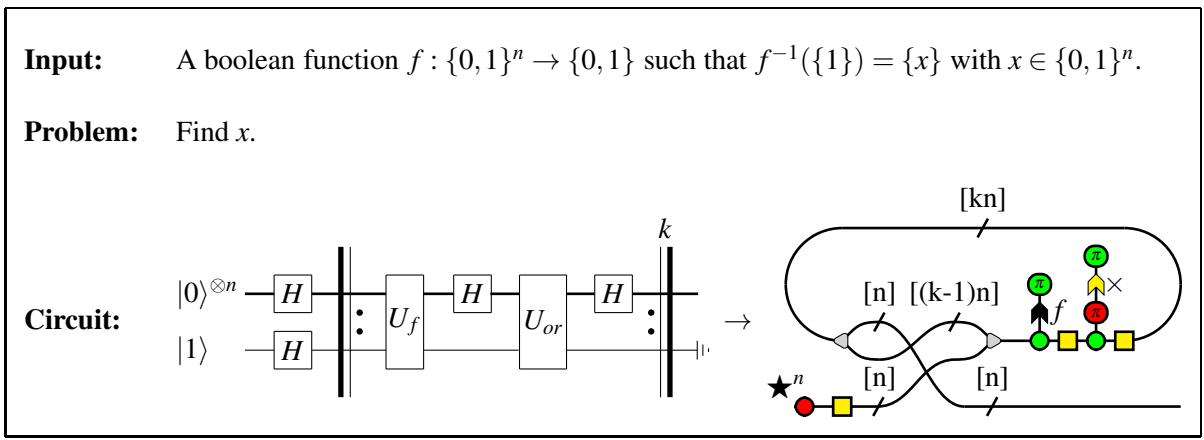
\begin{center}
	\begin{tabular}{|ll|}
		\hline
		&\\
		\textbf{Input:}& A boolean function $f:\{0,1\}^n \to \{0,1\}$ such that $f^{-1}(\{1\})=\{x\}$ with $x\in \{0,1\}^n $.\\
		&\\
		\textbf{Problem:}& Find $x$.\\
		&\\
		\textbf{Circuit:}& $\input{grov0.tikz} ~\to~ \input{grov1.tikz}$\\
		\hline
	\end{tabular}
\end{center}

Here we need to explain the translation into diagrams. First the ancillas is only here to form the diagonal oracles we then have:

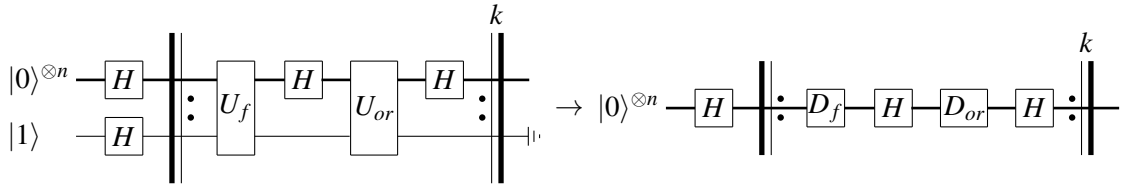
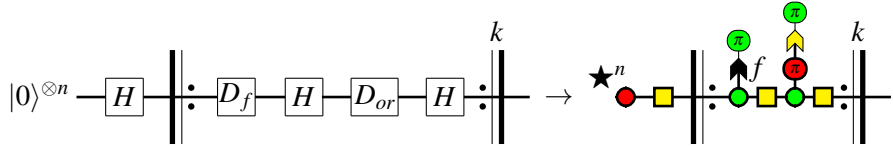
\begin{center}
	$\input{grov0.tikz} ~\to~ \input{grov2.tikz}$
\end{center}

translating into diagrams:

\begin{center}
	$\input{grov2.tikz} ~\to~ \input{grov3.tikz}$
\end{center}

Now, using the iteration mechanism of Lemma $3$ to represent the $k$ queries gives:
 
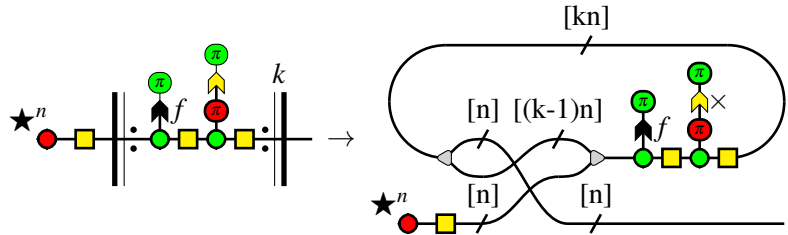
\begin{center}
	$\input{grov3.tikz} ~\to~ \input{grov1.tikz}$
\end{center}
 
The $\times$ stands for the matrix resulting of the thickening of the AND gate. The promise translates to: $\input{farrow0.tikz}=\input{grarrow.tikz}$ where $\bar x$ is the bit-wise negation of $x$. Our goal is to compute the probability of the outcome $\ket{x}$. Making the $x$ red phase slides gives:

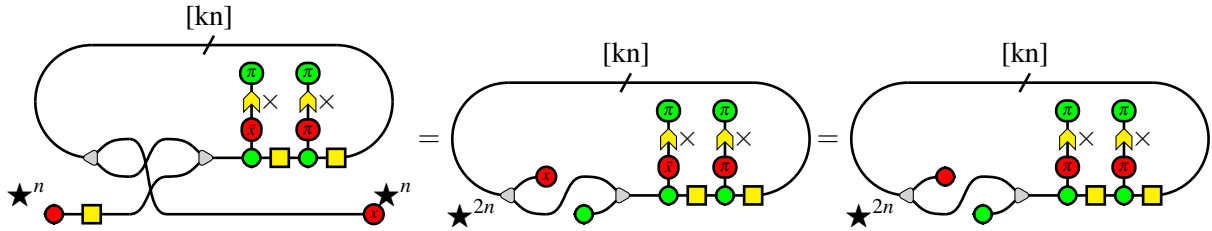
\begin{center}
$\input{grov4.tikz}=\input{grov5.tikz}=\input{grov6.tikz}$
\end{center}

Using the iteration mechanism we get: $\input{grov7.tikz}$.

Here we will translate a geometric approach into diagrams.

\begin{lemma}
	Setting $\nu\df\frac{1}{\sqrt{2^n -1}}$, $\cos(\frac{\mu}{2})\df \frac{-1}{\sqrt{2^n}}$ and $\sin(\frac{\mu}{2})\df \frac{\sqrt{2^n -1}}{\sqrt{2^n}}$, the map $\input{V.tikz}\df\input{proj0.tikz}$ satisfies:
	
	\begin{center}
		\begin{tabular}{ll}
			$\bullet \quad\begin{tikzpicture}
	\begin{pgfonlayer}{nodelayer}
		\node [style=grdot] (26) at (1, 0) {};
		\node [style=none] (28) at (0, 0) {};
	\end{pgfonlayer}
	\begin{pgfonlayer}{edgelayer}
		\draw [style=very thick] (26) to (28.center);
	\end{pgfonlayer}
\end{tikzpicture}
=\input{keto0.tikz}$ & $\bullet \quad\input{iso0.tikz}=\begin{tikzpicture}
	\begin{pgfonlayer}{nodelayer}
		\node [style=none] (240) at (2, 0) {};
		\node [style=none] (248) at (0, 0) {};
		\node [style=none] (251) at (1, 0) {};
	\end{pgfonlayer}
	\begin{pgfonlayer}{edgelayer}
		\draw (248.center) to (251.center);
		\draw (251.center) to (240.center);
	\end{pgfonlayer}
\end{tikzpicture}
$\\[0.5cm]
			
			$\bullet \quad\begin{tikzpicture}
	\begin{pgfonlayer}{nodelayer}
		\node [style=none] (12) at (1, 0) {};
		\node [style=none] (13) at (0, 0) {};
		\node [style=ggdot] (24) at (0, 0) {};
	\end{pgfonlayer}
	\begin{pgfonlayer}{edgelayer}
		\draw [style=very thick] (12.center) to (13.center);
	\end{pgfonlayer}
\end{tikzpicture}
=\input{keti3.tikz}$ & $\bullet \quad\input{stab0.tikz}=\input{staab.tikz}$
		\end{tabular}
	\end{center}

\end{lemma}

This lemma allows us to rewrite the diagram as follows: $\input{grov9.tikz}$.

Making the isometries slide gives:

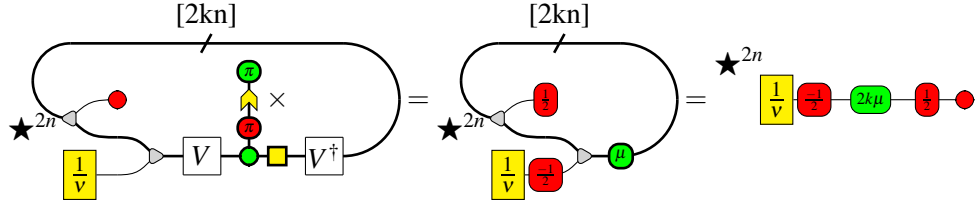
\begin{center}
	$\input{grov10.tikz}=\input{grov11.tikz}=\input{grov12.tikz}$
\end{center}

The last step being the iteration mechanism. We can now compute the interpretation: 

\begin{center}
	$\interp{\input{grov12.tikz}}= \left|\begin{pmatrix} 1&0 \end{pmatrix}\begin{pmatrix}
	\cos(k\mu )&\sin(k\mu )\\
	-\sin(k\mu )&\cos(k\mu )
	\end{pmatrix}\begin{pmatrix}-\cos(\frac{\mu }{2})\\\sin(\frac{\mu }{2})\end{pmatrix}\right|^2=\cos^2(\frac{2k+1}2\mu)$
\end{center}

So the probability is maximal when $\frac{2k+1}2\mu\simeq 0\bmod \pi$. Moreover $ \mu= \pi \pm \frac{2}{\sqrt{2^n}}+o(\frac{1}{\sqrt{2^n}})$, thus $k\simeq \frac{\pi}{4}\sqrt{2^n}$.

\section*{Conclusion}

In this article, we have used scalable notations to verify some standard quantum algorithms. For the moment,  this work is merely exploratory. Trying to tackle graphically many algorithms and protocols is the only way to evaluate the current ergonomics of graphical methods. The ultimate goal is to be able to compile high level quantum programming languages directly into diagrams. Then a graphical proof assistant could be used to provide proofs of correctness and optimizations. Case studies like in this paper are steps toward a double understanding. First, how graphical languages must be designed to fit this purpose. Second, what should be the specifications of future graphical proof assistants. Those two perspectives are clearly entangled. 

If the graphical verification of most of the algorithms we presented are neat and straightforward. Our approach of Grover's algorithm is still not rigorous enough to be implemented in a future proof assistant. More works need to be done on the higher level structure like the iteration mechanism we have  sketched out.


\bibliographystyle{eptcs}

\bibliography{scal.bib}

\begin{thebibliography}{10}
\providecommand{\bibitemdeclare}[2]{}
\providecommand{\surnamestart}{}
\providecommand{\surnameend}{}
\providecommand{\urlprefix}{Available at }
\providecommand{\url}[1]{\texttt{#1}}
\providecommand{\href}[2]{\texttt{#2}}
\providecommand{\urlalt}[2]{\href{#1}{#2}}
\providecommand{\doi}[1]{doi:\urlalt{http://dx.doi.org/#1}{#1}}
\providecommand{\bibinfo}[2]{#2}

\bibitemdeclare{inproceedings}{backens2018zh}
\bibitem{backens2018zh}
\bibinfo{author}{Miriam \surnamestart Backens\surnameend} \&
  \bibinfo{author}{Aleks \surnamestart Kissinger\surnameend}
  (\bibinfo{year}{2019}): \emph{\bibinfo{title}{{ZH}: A Complete Graphical
  Calculus for Quantum Computations Involving Classical Non-linearity}}.
\newblock In \bibinfo{editor}{Peter \surnamestart Selinger\surnameend} \&
  \bibinfo{editor}{Giulio \surnamestart Chiribella\surnameend}, editors: {\sl
  \bibinfo{booktitle}{{\rm Proceedings of the 15th International Conference on}
  Quantum Physics and Logic, {\rm Halifax, Canada, 3-7th June 2018}}}, {\sl
  \bibinfo{series}{Electronic Proceedings in Theoretical Computer Science}}
  \bibinfo{volume}{287}, \bibinfo{publisher}{Open Publishing Association}, pp.
  \bibinfo{pages}{23--42}, \doi{10.4204/EPTCS.287.2}.

\bibitemdeclare{article}{de2020well}
\bibitem{de2020well}
\bibinfo{author}{Niel \surnamestart de~Beaudrap\surnameend}
  (\bibinfo{year}{2020}): \emph{\bibinfo{title}{Well-tempered ZX and ZH
  calculi}}.
\newblock {\sl \bibinfo{journal}{arXiv preprint arXiv:2006.02557}}.

\bibitemdeclare{inproceedings}{de2020fast}
\bibitem{de2020fast}
\bibinfo{author}{Niel \surnamestart de~Beaudrap\surnameend},
  \bibinfo{author}{Xiaoning \surnamestart Bian\surnameend} \&
  \bibinfo{author}{Quanlong \surnamestart Wang\surnameend}
  (\bibinfo{year}{2020}): \emph{\bibinfo{title}{Fast and Effective Techniques
  for T-Count Reduction via Spider Nest Identities}}.
\newblock In: {\sl \bibinfo{booktitle}{15th Conference on the Theory of Quantum
  Computation, Communication and Cryptography}},
  \doi{10.4230/LIPIcs.TQC.2020.11}.

\bibitemdeclare{inproceedings}{Pauli-Fusion}
\bibitem{Pauli-Fusion}
\bibinfo{author}{Niel \surnamestart de~Beaudrap\surnameend},
  \bibinfo{author}{Ross \surnamestart Duncan\surnameend},
  \bibinfo{author}{Dominic \surnamestart Horsman\surnameend} \&
  \bibinfo{author}{Simon \surnamestart Perdrix\surnameend}
  (\bibinfo{year}{2020}): \emph{\bibinfo{title}{Pauli Fusion: a Computational
  Model to Realise Quantum Transformations from ZX Terms}}.
\newblock In \bibinfo{editor}{Bob \surnamestart Coecke\surnameend} \&
  \bibinfo{editor}{Matthew \surnamestart Leifer\surnameend}, editors: {\sl
  \bibinfo{booktitle}{{\rm Proceedings 16th International Conference on}
  Quantum Physics and Logic, {\rm Chapman University, Orange, CA, USA., 10-14
  June 2019}}}, {\sl \bibinfo{series}{Electronic Proceedings in Theoretical
  Computer Science}} \bibinfo{volume}{318}, \bibinfo{publisher}{Open Publishing
  Association}, pp. \bibinfo{pages}{85--105}, \doi{10.4204/EPTCS.318.6}.

\bibitemdeclare{article}{de2017zx}
\bibitem{de2017zx}
\bibinfo{author}{Niel \surnamestart de~Beaudrap\surnameend} \&
  \bibinfo{author}{Dominic \surnamestart Horsman\surnameend}
  (\bibinfo{year}{2017}): \emph{\bibinfo{title}{The ZX calculus is a language
  for surface code lattice surgery}}.
\newblock {\sl
  \bibinfo{journal}{https://quantum-journal.org/papers/q-2020-01-09-218/}},
  \doi{10.22331/q-2020-01-09-218}.

\bibitemdeclare{article}{bernstein1997quantum}
\bibitem{bernstein1997quantum}
\bibinfo{author}{Ethan \surnamestart Bernstein\surnameend} \&
  \bibinfo{author}{Umesh \surnamestart Vazirani\surnameend}
  (\bibinfo{year}{1997}): \emph{\bibinfo{title}{Quantum complexity theory}}.
\newblock {\sl \bibinfo{journal}{SIAM Journal on computing}}
  \bibinfo{volume}{26}(\bibinfo{number}{5}), pp. \bibinfo{pages}{1411--1473},
  \doi{10.1145/167088.167097}.

\bibitemdeclare{inproceedings}{carette2019szx}
\bibitem{carette2019szx}
\bibinfo{author}{Titouan \surnamestart Carette\surnameend},
  \bibinfo{author}{Dominic \surnamestart Horsman\surnameend} \&
  \bibinfo{author}{Simon \surnamestart Perdrix\surnameend}
  (\bibinfo{year}{2019}): \emph{\bibinfo{title}{SZX-Calculus: Scalable
  Graphical Quantum Reasoning}}.
\newblock In: {\sl \bibinfo{booktitle}{44th International Symposium on
  Mathematical Foundations of Computer Science (MFCS 2019)}},
  \bibinfo{organization}{Schloss Dagstuhl-Leibniz-Zentrum fuer Informatik},
  \doi{10.4230/LIPIcs.MFCS.2019.55}.

\bibitemdeclare{inproceedings}{carette2019completeness}
\bibitem{carette2019completeness}
\bibinfo{author}{Titouan \surnamestart Carette\surnameend},
  \bibinfo{author}{Emmanuel \surnamestart Jeandel\surnameend},
  \bibinfo{author}{Simon \surnamestart Perdrix\surnameend} \&
  \bibinfo{author}{Renaud \surnamestart Vilmart\surnameend}
  (\bibinfo{year}{2019}): \emph{\bibinfo{title}{Completeness of Graphical
  Languages for Mixed States Quantum Mechanics}}.
\newblock In: {\sl \bibinfo{booktitle}{International Colloquium on Automata,
  Languages, and Programming (ICALP'19)}}, \doi{10.4230/LIPIcs.ICALP.2019.108}.

\bibitemdeclare{article}{carette2020colored}
\bibitem{carette2020colored}
\bibinfo{author}{Titouan \surnamestart Carette\surnameend} \&
  \bibinfo{author}{Simon \surnamestart Perdrix\surnameend}
  (\bibinfo{year}{2020}): \emph{\bibinfo{title}{Colored props for large scale
  graphical reasoning}}.
\newblock {\sl \bibinfo{journal}{arXiv preprint arXiv:2007.03564}}.

\bibitemdeclare{article}{chancellor2016graphical}
\bibitem{chancellor2016graphical}
\bibinfo{author}{Nicholas \surnamestart Chancellor\surnameend},
  \bibinfo{author}{Aleks \surnamestart Kissinger\surnameend},
  \bibinfo{author}{Joschka \surnamestart Roffe\surnameend},
  \bibinfo{author}{Stefan \surnamestart Zohren\surnameend} \&
  \bibinfo{author}{Dominic \surnamestart Horsman\surnameend}
  (\bibinfo{year}{2016}): \emph{\bibinfo{title}{Graphical structures for design
  and verification of quantum error correction}}.
\newblock {\sl \bibinfo{journal}{arXiv preprint arXiv:1611.08012}}.

\bibitemdeclare{inproceedings}{chantawibul2018monoidal}
\bibitem{chantawibul2018monoidal}
\bibinfo{author}{Apiwat \surnamestart Chantawibul\surnameend} \&
  \bibinfo{author}{Pawe{\l} \surnamestart Soboci{\'n}ski\surnameend}
  (\bibinfo{year}{2018}): \emph{\bibinfo{title}{Monoidal Multiplexing}}.
\newblock In: {\sl \bibinfo{booktitle}{International Colloquium on Theoretical
  Aspects of Computing}}, \bibinfo{organization}{Springer}, pp.
  \bibinfo{pages}{116--131}, \doi{10.1007/978-3-030-02508-3_7}.

\bibitemdeclare{article}{coecke2011interacting}
\bibitem{coecke2011interacting}
\bibinfo{author}{Bob \surnamestart Coecke\surnameend} \& \bibinfo{author}{Ross
  \surnamestart Duncan\surnameend}: \emph{\bibinfo{title}{Interacting quantum
  observables: categorical algebra and diagrammatics}}.
\newblock {\sl \bibinfo{journal}{New Journal of Physics}}
  \bibinfo{volume}{13}(\bibinfo{number}{4}), p. \bibinfo{pages}{043016},
  \doi{10.1007/978-3-540-70583-3_25], year = {2011}}.

\bibitemdeclare{book}{picturing-qp}
\bibitem{picturing-qp}
\bibinfo{author}{Bob \surnamestart Coecke\surnameend} \& \bibinfo{author}{Aleks
  \surnamestart Kissinger\surnameend} (\bibinfo{year}{2017}):
  \emph{\bibinfo{title}{Picturing Quantum Processes: A First Course in Quantum
  Theory and Diagrammatic Reasoning}}.
\newblock \bibinfo{publisher}{Cambridge University Press},
  \doi{10.1017/9781316219317}.

\bibitemdeclare{article}{deutsch1992rapid}
\bibitem{deutsch1992rapid}
\bibinfo{author}{David \surnamestart Deutsch\surnameend} \&
  \bibinfo{author}{Richard \surnamestart Jozsa\surnameend}
  (\bibinfo{year}{1992}): \emph{\bibinfo{title}{Rapid solution of problems by
  quantum computation}}.
\newblock {\sl \bibinfo{journal}{Proceedings of the Royal Society of London.
  Series A: Mathematical and Physical Sciences}}
  \bibinfo{volume}{439}(\bibinfo{number}{1907}), pp. \bibinfo{pages}{553--558},
  \doi{10.1098/rspa.1992.0167}.

\bibitemdeclare{article}{duncan2019graph}
\bibitem{duncan2019graph}
\bibinfo{author}{Ross \surnamestart Duncan\surnameend}, \bibinfo{author}{Aleks
  \surnamestart Kissinger\surnameend}, \bibinfo{author}{Simon \surnamestart
  Pedrix\surnameend} \& \bibinfo{author}{John \surnamestart van~de
  Wetering\surnameend} (\bibinfo{year}{2019}):
  \emph{\bibinfo{title}{Graph-theoretic Simplification of Quantum Circuits with
  the ZX-calculus}}.
\newblock {\sl \bibinfo{journal}{arXiv preprint arXiv:1902.03178}},
  \doi{10.22331/q-2020-06-04-279}.

\bibitemdeclare{article}{gidney2018efficient}
\bibitem{gidney2018efficient}
\bibinfo{author}{Craig \surnamestart Gidney\surnameend} \&
  \bibinfo{author}{Austin~G \surnamestart Fowler\surnameend}
  (\bibinfo{year}{2018}): \emph{\bibinfo{title}{Efficient magic state factories
  with a catalyzed |CCZ> to 2|T> transformation}}.
\newblock {\sl \bibinfo{journal}{arXiv preprint arXiv:1812.01238}},
  \doi{10.22331/q-2019-04-30-135}.

\bibitemdeclare{article}{gogioso2017fully}
\bibitem{gogioso2017fully}
\bibinfo{author}{Stefano \surnamestart Gogioso\surnameend} \&
  \bibinfo{author}{Aleks \surnamestart Kissinger\surnameend}
  (\bibinfo{year}{2017}): \emph{\bibinfo{title}{Fully graphical treatment of
  the quantum algorithm for the Hidden Subgroup Problem}}.
\newblock {\sl \bibinfo{journal}{arXiv preprint arXiv:1701.08669}}.

\bibitemdeclare{article}{Gro97b}
\bibitem{Gro97b}
\bibinfo{author}{L.~K. \surnamestart Grover\surnameend} (\bibinfo{year}{1997}):
  \emph{\bibinfo{title}{Quantum Mechanics Helps in Searching for a Needle in a
  Haystack}}.
\newblock {\sl \bibinfo{journal}{Phys. Rev. Lett.}} \bibinfo{volume}{79}, p.
  \bibinfo{pages}{325}, \doi{10.1103/PhysRevLett.79.325}.

\bibitemdeclare{article}{hadzihasanovic2018zw}
\bibitem{hadzihasanovic2018zw}
\bibinfo{author}{Amar \surnamestart Hadzihasanovic\surnameend}
  (\bibinfo{year}{2018}): \emph{\bibinfo{title}{ZW calculi: diagrammatic
  languages for pure-state quantum computing}}.
\newblock {\sl \bibinfo{journal}{Logic and Applications LAP 2018}},
  p.~\bibinfo{pages}{13}.

\bibitemdeclare{inproceedings}{HNW}
\bibitem{HNW}
\bibinfo{author}{Amar \surnamestart Hadzihasanovic\surnameend},
  \bibinfo{author}{Kang~Feng \surnamestart Ng\surnameend} \&
  \bibinfo{author}{Quanlong \surnamestart Wang\surnameend}
  (\bibinfo{year}{2018}): \emph{\bibinfo{title}{Two Complete Axiomatisations of
  Pure-state Qubit Quantum Computing}}.
\newblock In: {\sl \bibinfo{booktitle}{Proceedings of the 33rd Annual ACM/IEEE
  Symposium on Logic in Computer Science}}, \bibinfo{series}{LICS '18},
  \bibinfo{publisher}{ACM}, \bibinfo{address}{New York, NY, USA}, pp.
  \bibinfo{pages}{502--511}, \doi{10.1145/3209108.3209128}.

\bibitemdeclare{article}{hanks2020effective}
\bibitem{hanks2020effective}
\bibinfo{author}{Michael \surnamestart Hanks\surnameend},
  \bibinfo{author}{Marta~P \surnamestart Estarellas\surnameend},
  \bibinfo{author}{William~J \surnamestart Munro\surnameend} \&
  \bibinfo{author}{Kae \surnamestart Nemoto\surnameend} (\bibinfo{year}{2020}):
  \emph{\bibinfo{title}{Effective Compression of Quantum Braided Circuits Aided
  by ZX-Calculus}}.
\newblock {\sl \bibinfo{journal}{Physical Review X}}
  \bibinfo{volume}{10}(\bibinfo{number}{4}), p. \bibinfo{pages}{041030},
  \doi{10.1103/PhysRevX.10.041030}.

\bibitemdeclare{inproceedings}{jeandel2018complete}
\bibitem{jeandel2018complete}
\bibinfo{author}{Emmanuel \surnamestart Jeandel\surnameend},
  \bibinfo{author}{Simon \surnamestart Perdrix\surnameend} \&
  \bibinfo{author}{Renaud \surnamestart Vilmart\surnameend}
  (\bibinfo{year}{2018}): \emph{\bibinfo{title}{A complete axiomatisation of
  the {ZX}-calculus for {C}lifford+{T} quantum mechanics}}.
\newblock In: {\sl \bibinfo{booktitle}{Proceedings of the 33rd Annual ACM/IEEE
  Symposium on Logic in Computer Science (LICS)}}, \bibinfo{organization}{ACM},
  pp. \bibinfo{pages}{559--568}, \doi{10.1145/3209108.3209131}.

\bibitemdeclare{inproceedings}{jeandel2018diagrammatic}
\bibitem{jeandel2018diagrammatic}
\bibinfo{author}{Emmanuel \surnamestart Jeandel\surnameend},
  \bibinfo{author}{Simon \surnamestart Perdrix\surnameend} \&
  \bibinfo{author}{Renaud \surnamestart Vilmart\surnameend}
  (\bibinfo{year}{2018}): \emph{\bibinfo{title}{Diagrammatic reasoning beyond
  Clifford+ T quantum mechanics}}.
\newblock In: {\sl \bibinfo{booktitle}{Proceedings of the 33rd Annual ACM/IEEE
  Symposium on Logic in Computer Science (LICS)}}, \bibinfo{organization}{ACM},
  pp. \bibinfo{pages}{569--578}, \doi{10.1145/3209108.3209139}.

\bibitemdeclare{article}{kissinger2019reducing}
\bibitem{kissinger2019reducing}
\bibinfo{author}{Aleks \surnamestart Kissinger\surnameend} \&
  \bibinfo{author}{John \surnamestart van~de Wetering\surnameend}
  (\bibinfo{year}{2019}): \emph{\bibinfo{title}{Reducing T-count with the
  ZX-calculus}}.
\newblock {\sl \bibinfo{journal}{arXiv preprint arXiv:1903.10477}}.

\bibitemdeclare{misc}{nielsen2002quantum}
\bibitem{nielsen2002quantum}
\bibinfo{author}{Michael~A \surnamestart Nielsen\surnameend} \&
  \bibinfo{author}{Isaac \surnamestart Chuang\surnameend}
  (\bibinfo{year}{2002}): \emph{\bibinfo{title}{Quantum computation and quantum
  information}}.

\bibitemdeclare{article}{simon1997power}
\bibitem{simon1997power}
\bibinfo{author}{Daniel~R \surnamestart Simon\surnameend}
  (\bibinfo{year}{1997}): \emph{\bibinfo{title}{On the power of quantum
  computation}}.
\newblock {\sl \bibinfo{journal}{SIAM journal on computing}}
  \bibinfo{volume}{26}(\bibinfo{number}{5}), pp. \bibinfo{pages}{1474--1483},
  \doi{10.1137/S0097539796298637}.

\bibitemdeclare{inproceedings}{vicary2013topological}
\bibitem{vicary2013topological}
\bibinfo{author}{Jamie \surnamestart Vicary\surnameend} (\bibinfo{year}{2013}):
  \emph{\bibinfo{title}{Topological structure of quantum algorithms}}.
\newblock In: {\sl \bibinfo{booktitle}{2013 28th Annual ACM/IEEE Symposium on
  Logic in Computer Science}}, \bibinfo{organization}{IEEE}, pp.
  \bibinfo{pages}{93--102}, \doi{10.1109/LICS.2013.14}.

\bibitemdeclare{article}{zeng2014abstract}
\bibitem{zeng2014abstract}
\bibinfo{author}{William \surnamestart Zeng\surnameend} \&
  \bibinfo{author}{Jamie \surnamestart Vicary\surnameend}
  (\bibinfo{year}{2014}): \emph{\bibinfo{title}{Abstract structure of unitary
  oracles for quantum algorithms}}.
\newblock {\sl \bibinfo{journal}{arXiv preprint arXiv:1406.1278}},
  \doi{10.4204/EPTCS.172.19}.

\end{thebibliography}

\appendix

\section{Proofs}

\begin{lemma}[Iteration]
	Given any diagram $f:a\to a$ in $S\mathcal{D}$:
	
	\begin{center}
		$\input{iter0.tikz}=\left(\input{iter1.tikz}\right)^{k+1}$
	\end{center}
	
\end{lemma}

\begin{proof}
	By induction, for $k=0$: 
	
	\begin{center}
		$\input{iter2.tikz}=\input{iter3.tikz}=\input{iter4.tikz}=\input{iter1.tikz}$
	\end{center}
	
	For $k>0$, let $k=l+1$:
	
	\begin{center}
		$\input{iter5.tikz}=\input{iter6.tikz}=\input{iter7.tikz}=\input{iter8.tikz}=\input{iter9.tikz}=\input{iter10.tikz}$
	\end{center}
	
\end{proof}

\begin{lemma}
	Given a function arrow $f$:
	
	\begin{center}
		\begin{tabular}{c}
			$f$ is \textbf{balanced} $\quad\stackrel{def}{\Leftrightarrow}\quad\forall x,y\in 2^m ~ |f^{-1}\left(\{x\}\right)|=|f^{-1}\left(\{y\}\right)|\quad\Leftrightarrow\quad\input{fgcoerase0.tikz}=\begin{tikzpicture}
	\begin{pgfonlayer}{nodelayer}
		\node [style=ggdot] (69) at (0, 0) {};
		\node [style=none] (71) at (1, 0) {};
	\end{pgfonlayer}
	\begin{pgfonlayer}{edgelayer}
		\draw [style=very thick] (71.center) to (69);
	\end{pgfonlayer}
\end{tikzpicture}
$ \\
			$f$ is \textbf{injective} $\quad\stackrel{def}{\Leftrightarrow}\quad\forall x,y\in 2^n ~ \left(f(x)=f(y) \Rightarrow x=y\right)\quad\Leftrightarrow\quad\input{fgcocopy0.tikz}=\input{fgcocopy1.tikz}$
		\end{tabular}
	\end{center}
	
\end{lemma}

\begin{proof}
	
	\begin{align*}
	\input{fgcoerase0.tikz}=\begin{tikzpicture}
	\begin{pgfonlayer}{nodelayer}
		\node [style=ggdot] (69) at (0, 0) {};
		\node [style=none] (71) at (1, 0) {};
	\end{pgfonlayer}
	\begin{pgfonlayer}{edgelayer}
		\draw [style=very thick] (71.center) to (69);
	\end{pgfonlayer}
\end{tikzpicture}
&\quad\Leftrightarrow\quad\interp{\input{fgcoerase0.tikz}}=\interp{\begin{tikzpicture}
	\begin{pgfonlayer}{nodelayer}
		\node [style=ggdot] (69) at (0, 0) {};
		\node [style=none] (71) at (1, 0) {};
	\end{pgfonlayer}
	\begin{pgfonlayer}{edgelayer}
		\draw [style=very thick] (71.center) to (69);
	\end{pgfonlayer}
\end{tikzpicture}
}\\
	&\quad\Leftrightarrow\quad  \sum\limits_{x\in \textbf{2}^n} \ket{f(x)}=\sum\limits_{y\in \textbf{2}^m} 2^{\frac{n-m}{2}}\ket{y}\\
	&\quad\Leftrightarrow\quad \forall x,y\in 2^m ~ |f^{-1}\left(\{x\}\right)|=|f^{-1}\left(\{y\}\right)|
	\end{align*}
	
	\begin{align*}
	\input{fgcocopy0.tikz}=\input{fgcocopy1.tikz}&\quad\Leftrightarrow\quad\interp{\input{fgcocopy0.tikz}}=\interp{\input{fgcocopy1.tikz}}\\
	&\quad\Leftrightarrow\quad  \forall x,y\in 2^n ~\delta_{x=y} \ket{f(x)}=\delta_{f(x)=f(y)} \ket{f(x)}\\
	&\quad\Leftrightarrow\quad \forall x,y\in 2^n ~ \left(f(x)=f(y) \Rightarrow x=y\right)
	\end{align*}
	
\end{proof}



\begin{lemma}
	Setting $\nu\df\frac{1}{\sqrt{2^n -1}}$, $\cos(\frac{\mu \pi}{2})\df \frac{-1}{\sqrt{2^n}}$ and $\sin(\frac{\mu \pi}{2})\df \frac{\sqrt{2^n -1}}{\sqrt{2^n}}$, the map $\input{V.tikz}\df\input{proj0.tikz}$ satisfies:
	
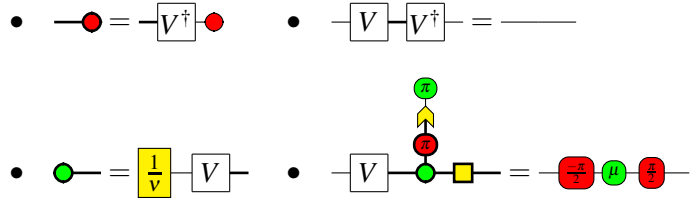
	\begin{center}
		\begin{tabular}{ll}
			$\bullet \quad\begin{tikzpicture}
	\begin{pgfonlayer}{nodelayer}
		\node [style=grdot] (26) at (1, 0) {};
		\node [style=none] (28) at (0, 0) {};
	\end{pgfonlayer}
	\begin{pgfonlayer}{edgelayer}
		\draw [style=very thick] (26) to (28.center);
	\end{pgfonlayer}
\end{tikzpicture}
=\input{keto0.tikz}$ & $\bullet \quad\input{iso0.tikz}=\begin{tikzpicture}
	\begin{pgfonlayer}{nodelayer}
		\node [style=none] (240) at (2, 0) {};
		\node [style=none] (248) at (0, 0) {};
		\node [style=none] (251) at (1, 0) {};
	\end{pgfonlayer}
	\begin{pgfonlayer}{edgelayer}
		\draw (248.center) to (251.center);
		\draw (251.center) to (240.center);
	\end{pgfonlayer}
\end{tikzpicture}
$\\[0.5cm]
			
			$\bullet \quad\begin{tikzpicture}
	\begin{pgfonlayer}{nodelayer}
		\node [style=none] (12) at (1, 0) {};
		\node [style=none] (13) at (0, 0) {};
		\node [style=ggdot] (24) at (0, 0) {};
	\end{pgfonlayer}
	\begin{pgfonlayer}{edgelayer}
		\draw [style=very thick] (12.center) to (13.center);
	\end{pgfonlayer}
\end{tikzpicture}
=\input{keti3.tikz}$ & $\bullet \quad\input{stab0.tikz}=\input{staab.tikz}$
		\end{tabular}
	\end{center}
	
\end{lemma}

\begin{proof} We proceed point by point:
	
	\begin{itemize}
		
		\item By rewriting:
		
		\begin{center}
			$\input{keto0.tikz}=\input{keto1.tikz}=\input{keto2.tikz}=\begin{tikzpicture}
	\begin{pgfonlayer}{nodelayer}
		\node [style=grdot] (26) at (1, 0) {};
		\node [style=none] (28) at (0, 0) {};
	\end{pgfonlayer}
	\begin{pgfonlayer}{edgelayer}
		\draw [style=very thick] (26) to (28.center);
	\end{pgfonlayer}
\end{tikzpicture}
$
		\end{center}
		
		\item By rewriting:
		
		\begin{center}
			$\begin{tikzpicture}
	\begin{pgfonlayer}{nodelayer}
		\node [style=none] (12) at (1, 0) {};
		\node [style=none] (13) at (0, 0) {};
		\node [style=ggdot] (24) at (0, 0) {};
	\end{pgfonlayer}
	\begin{pgfonlayer}{edgelayer}
		\draw [style=very thick] (12.center) to (13.center);
	\end{pgfonlayer}
\end{tikzpicture}
=\input{keti1.tikz}=\input{keti2.tikz}=\input{keti3.tikz}$
		\end{center}
		
		\item First we can check: $\interp{\input{defx1.tikz}}=\interp{\begin{tikzpicture}
	\begin{pgfonlayer}{nodelayer}
		\node [style=none] (251) at (1, 0) {};
		\node [style=gdot] (254) at (0, 0) {};
	\end{pgfonlayer}
	\begin{pgfonlayer}{edgelayer}
		\draw (251.center) to (254);
	\end{pgfonlayer}
\end{tikzpicture}
}$ and then $\input{defx1.tikz}=\begin{tikzpicture}
	\begin{pgfonlayer}{nodelayer}
		\node [style=none] (251) at (1, 0) {};
		\node [style=gdot] (254) at (0, 0) {};
	\end{pgfonlayer}
	\begin{pgfonlayer}{edgelayer}
		\draw (251.center) to (254);
	\end{pgfonlayer}
\end{tikzpicture}
$. So:
		\begin{center}
			$\input{iso0.tikz}=\input{proj1.tikz}=\input{proj2.tikz}=\input{proj3.tikz}=\begin{tikzpicture}
	\begin{pgfonlayer}{nodelayer}
		\node [style=none] (240) at (2, 0) {};
		\node [style=none] (248) at (0, 0) {};
		\node [style=gdot] (251) at (1, 0) {};
		\node [style=gdot] (254) at (1, 1) {};
	\end{pgfonlayer}
	\begin{pgfonlayer}{edgelayer}
		\draw (248.center) to (251);
		\draw (251) to (240.center);
		\draw (251) to (254);
	\end{pgfonlayer}
\end{tikzpicture}
=\begin{tikzpicture}
	\begin{pgfonlayer}{nodelayer}
		\node [style=none] (240) at (2, 0) {};
		\node [style=none] (248) at (0, 0) {};
		\node [style=none] (251) at (1, 0) {};
	\end{pgfonlayer}
	\begin{pgfonlayer}{edgelayer}
		\draw (248.center) to (251.center);
		\draw (251.center) to (240.center);
	\end{pgfonlayer}
\end{tikzpicture}
$
		\end{center}
	
		\item First we have : $\input{stab0.tikz}=\input{stab1.tikz}=\input{stab2.tikz}$. So the question reduces to show that: $\input{stab3.tikz}=\input{stab4.tikz}$.
		
		We proceed by induction on $n$. If $n=1$ then $\nu=1$ and $\mu=\frac{-\pi}{2}$:
		
		\begin{center}
			$\begin{tikzpicture}
	\begin{pgfonlayer}{nodelayer}
		\node [style=had] (65) at (1, 0) {};
		\node [style=none] (67) at (0, 0) {};
		\node [style=none] (74) at (2, 0) {};
	\end{pgfonlayer}
	\begin{pgfonlayer}{edgelayer}
		\draw (65) to (67.center);
		\draw (74.center) to (65);
	\end{pgfonlayer}
\end{tikzpicture}
=\input{stab6.tikz}=\input{stab7.tikz}=\input{stab8.tikz}$
		\end{center}
		
		For $n>1$ we need the following observation:
		
		\begin{center}
			 $\frac{1}{\sqrt{2}}\begin{pmatrix}
			1&1\\1&-1
			\end{pmatrix} \otimes \begin{pmatrix}
			\frac{1}{\sqrt{2^{n-1}}}&\frac{2^{n-1} -1}{\sqrt{2^{n-1}}}\\ \frac{1}{\sqrt{2^{n-1}}}& \frac{-1}{\sqrt{2^{n-1}}}
			\end{pmatrix}\begin{pmatrix}
			1&0\\
			0&1\\
			0&1\\
			0&1
			\end{pmatrix}=\begin{pmatrix}
			1&0\\
			0&1\\
			0&1\\
			0&1
			\end{pmatrix}\begin{pmatrix}
			\frac{1}{\sqrt{2^{n}}}&\frac{2^{n} -1}{\sqrt{2^{n}}}\\ \frac{1}{\sqrt{2^{n}}}& \frac{-1}{\sqrt{2^{n}}}
			\end{pmatrix}$.
		\end{center}
		
		 Graphically:
		 
		 \begin{center}
		 	$\input{stab9.tikz}=\input{stab10.tikz}$
		 \end{center}
		 
		 with $\nu'\df \frac{1}{\sqrt{2^{n-1}-1}}$, $\cos(\frac{\mu' \pi}{2})\df \frac{-1}{\sqrt{2^{n-1}}}$ and $\sin(\frac{\mu' \pi}{2})\df \frac{\sqrt{2^{n-1} -1}}{\sqrt{2^{n-1}}}$.
		 
		 We have:
		 
		 \begin{center}
		 	$\input{stab11.tikz}=\input{stab12.tikz}$
		 \end{center}
	 
		 The induction hypothesis gives:
		 
		 \begin{center}
		 	$\input{stab13.tikz}$
		 \end{center}
	 
		 And finally, using the previous observation: 
		 
		 \begin{center}
		 	$\input{stab14.tikz}=\input{stab15.tikz}$
		 \end{center}
		
	\end{itemize}

\end{proof}

\end{document}